\documentclass[letterpaper]{article} 
\usepackage[draft]{aaai2026}  
\usepackage{times}  
\usepackage{helvet}  
\usepackage{courier}  
\usepackage[hyphens]{url}  
\usepackage{graphicx} 
\usepackage{natbib}  
\usepackage{caption} 
\usepackage{algorithm}
\usepackage{algorithmic}

\usepackage{amsmath}
\usepackage{amssymb}

\usepackage{amsthm}

\newtheorem{definition}{Definition}
\newtheorem{example}{Example}

\theoremstyle{remark}
\newtheorem{claim}{Claim}
\newtheorem{step}{Step}
\newtheorem{case}{Case}
\usepackage{etoolbox}
\AtEndEnvironment{proof}{
    \setcounter{claim}{0}
    \setcounter{step}{0}
    \setcounter{case}{0}
}
\usepackage{amssymb}
\newenvironment{claimproof}[1][]
    {
        \noindent\textit{Proof of #1.\;}
    }
    { 
        \hfill $\vartriangleleft$\\[-2pt]
    }

\usepackage{mathtools}
\usepackage{cleveref}
\crefname{claim}{claim}{claims}
\crefname{step}{step}{Steps}
\crefname{case}{case}{Case}

\usepackage{thm-restate}

\usepackage{todonotes}
\usepackage{comment}

\usepackage{enumitem}
\usepackage{csquotes}

\newcommand{\N}{\mathbb{N}}
\newcommand{\T}{\mathcal{T}}
\newcommand{\E}{\mathbb{E}}

\newcommand{\s}{\mathbf{s}}
\newcommand{\U}{\mathcal{U}}

\newcommand{\sqda}{SQD}
\newcommand{\sbad}{\s_{\text{bad}}}

\usepackage{newfloat}
\usepackage{listings}
\DeclareCaptionStyle{ruled}{labelfont=normalfont,labelsep=colon,strut=off} 
\floatstyle{ruled}
\newfloat{listing}{tb}{lst}{}
\floatname{listing}{Listing}

\title{Optimal Welfare in Noncooperative Network Formation under Attack}
\author {
    Natan Doubez\textsuperscript{\rm 1},
    Pascal Lenzner\textsuperscript{\rm 2},
    Marcus Wunderlich\textsuperscript{\rm 2}
}
\affiliations {
    \textsuperscript{\rm 1}{\'E}cole Polytechnique, Palaiseau, France\\
    \textsuperscript{\rm 2}University of Augsburg, Augsburg, Germany\\
    natan.doubez@polytechnique.edu, \{pascal.lenzner,marcus.wunderlich\}@uni-a.de
}

\begin{document}

\maketitle

\begin{abstract}
Communication networks are essential for our economy and our everyday lives. This makes them lucrative targets for attacks. Today, we see an ongoing battle between criminals that try to disrupt our key communication networks and security professionals that try to mitigate these attacks. However, today's networks, like the Internet or peer-to-peer networks among smart devices, are not controlled by a single authority, but instead consist of many independently administrated entities that are interconnected. Thus, both the decisions of how to interconnect and how to secure against potential attacks are taken in a decentralized way by selfish agents.

This strategic setting, with agents that want to interconnect and potential attackers that want to disrupt the network, was captured via an influential game-theoretic model by Goyal, Jabbari, Kearns, Khanna, and Morgenstern (WINE 2016). We revisit this model and show improved tight bounds on the achieved robustness of networks created by selfish agents. As our main result, we show that such networks can resist attacks of a large class of potential attackers, i.e., these networks maintain asymptotically optimal welfare post attack. This improves several bounds and resolves an open problem. Along the way, we show the counter-intuitive result, that attackers that aim at minimizing the social welfare post attack do not actually inflict the greatest possible damage.  
\end{abstract}

\section{Introduction}
We rely on various networks: for communication, for transportation, and also for our energy infrastructure. Given their importance, networks have always been prone to attack and network operators have always tried to secure their networks against possible threats. A prominent example is the use of firewalls in routers to prevent a computer virus spreading in some part of the network from infecting other parts. Also, social distancing and vaccination measures in the past COVID pandemic can be understood as security measures in a (social) network to prevent the spread of a virus.

However, both examples show that in today's networks there is no central authority that could enforce certain security measures or a certain structure of the network. The use of a firewall or to vaccinate is an individual decision of the participants of the networks. The same holds for the decision of which connections to establish: while a direct connection yields benefit in terms of low latency, it also poses a risk since a possible attack could spread along the created link. Thus, today's networks can be better modeled and understood as a complex multi-agent system consisting of strategic smart agents. These agents can be people, routers, smart devices, or simply interacting components of an AI system. Each agent strategically decides on its security measures and on its links it wants to establish towards other agents.

This viewpoint of networks as multi-agent systems of strategic agents has sparked the multifaceted research on game-theoretic network formation models in Economics, Computer Science, and Artificial Intelligence in the last three decades, see e.g.~\cite{Pap01,jackson2008social}. In these models, agents that correspond to nodes of a network strategically decide which connections to other nodes to establish. The agents are selfish and act according to some given utility function that encodes the agents' objectives. In this paper, we focus on the objective of network robustness.
Modern communication and infrastructure networks have to cope with hardware failures or even deliberate attacks while still providing a reliable service. For incorporating this, researchers have studied agent-based network formation models where the agents prepare for single-link failures~\cite{BG03,Meirom15,CLMM16} or try to maximize the obtained min-cut in the created network~\cite{EchzellLM20}. Also the scenario of fighting a virus that spreads in the created network was considered and this is our main reference point. 

In this work we revisit the elegant strategic network formation model with attack and immunization by Goyal, Jabbari, Kearns, Khanna, and Morgenstern~\cite{Goyal16} that features a smart adversary. 
Each agent selfishly decides which costly links to form and if to acquire protection from attack. In such an attack, a single node of the network is targeted for infection and then this infection spreads along the subgraph consisting of unprotected nodes. The authors consider three natural types of attack: \begin{itemize}
\item \textbf{maximum carnage:} the attacker targets a node to infect as many nodes as possible,
\item \textbf{random attack:} the attacker targets an unprotected node uniformly at random, and
\item \textbf{maximum disruption:} the attacker targets a node to minimize the social welfare post attack. 
\end{itemize}
The utility of uninfected agents is the (expected) number of reachable uninfected nodes post attack, while infected agents have utility zero. The social welfare is the total utility. 

\citet{Goyal16} give non-trivial bounds on the social welfare of equilibrium networks for the maximum carnage and the random adversary and they pose the analysis of the maximum disruption attacker as open problem. In this paper, we completely resolve the question of the social welfare by providing tight optimal bounds for all three attackers. Even more, we show that the optimal bound holds for more general class of attackers that subsumes the maximum disruption attacker. In fact, we show that the created equilibrium networks have asymptotically the same social welfare post attack as in a setting without adversary. This highlights that networks that are created in a decentralized way by strategic agents are highly robust against various types of attackers.      

\subsection{Model and Preliminaries}

\textbf{Sets \& Functions:} We use $\N$ to denote the set of natural numbers and $\mathbb{R}^+$ to denote the set of non-negative real numbers. We will refer to the derivative of a discrete function $f : \N \rightarrow \mathbb{R}^+$ as the function $f' : n \mapsto f(n+1) - f(n)$.

\textbf{The Game:} For our network formation game, we mostly use the original notation by~\citet{Goyal16}. A game instance is defined by the tuple $(n,C_E,C_I,\mathcal{A})$, where $n$ is the number of agents (or nodes of the network), the value $C_E>1$ is the cost at which agents can buy an edge to any other node, the cost for a player to \emph{immunize} itself is $C_I>0$, and $\mathcal{A}$ describes the attacker (or opponent) targeting the created network. We use $[n] \coloneqq \{1,\dots,n\}$ as the set of agents and we use the terms agents and nodes interchangeably.

\textbf{Strategies:} Every agent can (1) buy arbitrarily many undirected incident edges to other agents, at a cost of $C_E$ per edge, and (2) immunize itself at a cost of $C_I$. The set of nodes to which an agent~$i \in [n]$ buys an edge is $X_i \subseteq [n]$, where $j \in X_i$ indicates that agent~$i$ buys the edge $\{i,j\}$. Whether agent~$i$ immunizes itself is represented via Boolean variable $y_i \in \{0,1\}$, i.e., $y_i=1$ if and only if agent~$i$ is immunized. If $y_i=0$, agent~$i$ is called \emph{vulnerable}. The pair $s_i \coloneqq (X_i,y_i)$ is the \emph{strategy} of agent~$i$, and a vector $\s = (s_1,\dots,s_n)$ of strategies of all agents is called a \emph{strategy profile}. Every strategy profile $\s$ induces an undirected network $G(\s)=(V,E(\s))$, with $V \coloneqq [n]$ being the set of agents, partitioned into immunized agents $\mathcal{I}(\s)$ and vulnerable agents $\mathcal{U}(\s)$. The set
$E(\s) = \{\{i,j\} \mid i\in X_j \vee j \in X_i\}$ is the set of bought edges. 
If it is clear from the context, we will omit the reference to the strategy profile~$\s$.

\textbf{Attacks:} Given a strategy profile $\s$ and its induced graph $G(\s)=(\mathcal{I}(\s)\cup\mathcal{U}(\s),E(\s))$, the attacker $\mathcal{A}$ will choose a single node as target. Attacking node~$v \in \mathcal{U}(\s)$ \emph{infects} its connected component in the network $G[\mathcal{U}(\s)]$ induced by the vulnerable agents (while attacking an immunized node has no impact). Infected nodes will be completely removed from the network. 
We call the connected components of $G[\mathcal{U}(\s)]$ \emph{vulnerable regions} and define $\mathcal{V}(\s)$ to be the set of all vulnerable regions of $G(\s)$. We similarly call a connected component of $G[\mathcal{I}(\s)]$ an \emph{immunized region}. Formally, an \emph{attacker} $\mathcal{A}$ is defined by a probability distribution over the vulnerable regions $\mathcal{V}(\s)$ that states how likely it is for $\mathcal{A}$ to attack each of them. This distribution heavily depends on the type of attacker, that we will describe in detail later. Given a fixed attacker $\mathcal{A}$, we refer to the vulnerable regions of $G(\s)$ that $\mathcal{A}$ attacks with non-zero probability as \emph{targeted regions} and define $\mathcal{T}(\s,\mathcal{A})$ as the set of all regions targeted by $\mathcal{A}$. All nodes inside a targeted region are called \emph{targeted nodes}. 

\textbf{Utilities:}
The utility of an agent~$i$ in profile $\s$  is a combination of its connectivity in network $G(\s)$ post attack minus the agent's costs for edges and immunization. Formally, the post attack connectivity of agent~$i$ is the expected size of $i$'s connected component after the attacker infected (and thus destroyed) a targeted region. To be precise, let $CC_i(T,\s)$, for some vulnerable region $T$, be the size of the connected component of $G[V \setminus T]$ that contains agent~$i$. Then, the \emph{connectivity} of agent~$i$ in strategy profile~$\s$ with attacker $\mathcal{A}$ is
    $$\textstyle\E_{\mathcal{T}(\s,\mathcal{A})}[CC_i(\s)]\coloneqq\sum_{T \in \mathcal{T}(\s,\mathcal{A})}\left(\mathbb{P}_\mathcal{A}[T,\s]\cdot CC_i(T,\s)\right),$$
where $\mathbb{P}_\mathcal{A}[T,\s]$ is the probability at which $\mathcal{A}$ attacks~$T$ in profile~$\s$.
Further, the \emph{utility} of agent~$i$ in this setting then is 
    $$u_i(\s) \coloneqq \E_{\mathcal{T}(\s,\mathcal{A})}[CC_i(\s)] - |X_i|C_E - y_iC_I,$$
where $|X_i|C_E + y_iC_I$ is the \emph{cost} of agent~$u$ in strategy profile~$\s$. The sum over the utilities of all agents is called the \emph{social welfare} of profile $\s$. We sometimes write $CC_i(v,\s)$ instead of $CC_i(T,\s)$, where $T$ is the targeted region containing ~$v \in \mathcal{U}$. And similarly $U_f(v,\s)$ for $U_f(T,\s)$.

\textbf{Equilibria:} Given strategy profile $\s$, we say that agent~$i$ \emph{plays best response}, if given the strategies of all other agents, agent~$i$ cannot deviate to a different strategy that yields strictly higher utility. If all agents play best response with respect to strategy profile $\s$, then we say that $\s$ is a \emph{Nash equilibrium}, and we say that $G(\s)$ is an \emph{equilibrium network}. We say that profile~$\s$ is a \emph{non-trivial} Nash equilibrium, if $G(\s)$ has at least one edge and if $|\mathcal{I}(\s)|>0$. These equilibria do not exist in the attack-free setting~\cite{bala2000noncooperative}.

\textbf{Different Types of Attackers:}
The simplest attacker is the one that attacks every vulnerable node uniformly at random. We refer to this one as \emph{random attack}. This attacker, and the two other attackers, \emph{maximum carnage} and \emph{maximum disruption}, introduced by~\citet{Goyal16}, belong to the class of \emph{well-behaved opponents} that have the same attack distribution for equivalent networks (see Definition~4 in~\cite{Goyal16} for details). We will consider a subclass of this, called $f$-opponents, that subsumes the maximum carnage and the maximum disruption opponents. For $f$-opponents, every attack is defined by a function $f\colon \{0,\dots,n\} \to \mathbb{R}^+$ that maps every size of a possible connected component of the network to some non-negative value, with\footnote{Setting $f(0)=0$ is not necessary but useful. E.g. when the opponent targets a region inside a connected component $Z$ that does not create a new component. The change in utility is the difference of the images of $f$ of the sizes of $Z$ before and after the attack. Setting $f(0)=0$ allows it to still hold if the attack completely deletes~$Z$. It is also useful when using the convexity of~$f$.} $f(0) = 0$. With this, an $f$-opponent $\mathcal{A}$ aims to minimize $U_f(T,\s) \coloneqq \sum_{K \in \mathcal{K}(T)}f(|K|)$, where $\mathcal{K}(T)$ is the set of connected components of the network after $\mathcal{A}$ attacked some vulnerable region $T \in \mathcal{T}(\s,\mathcal{A})$. Formally, the $f$-opponent attacks only vulnerable regions $T \in \mathcal{T}(\s,\mathcal{A})$ of $G(\s)$ that minimize $U_f(\cdot,\s)$ and chooses uniformly at random among them. In this framework, the \emph{maximum carnage} attacker and the \emph{maximum disruption} can be defined to be $x^1$- and $x^2$-opponents respectively, where $x^r\colon n \mapsto n^r$ is the monomial of degree $r$ for $r \in \N$.

\begin{example}\label{example}
To better understand the difference between the maximum carnage and maximum disruption attackers, consider the instance depicted in \Cref{fig:counter-greed-drop}. 
\begin{figure}[ht]
    \centering
    \includegraphics[width=\linewidth]{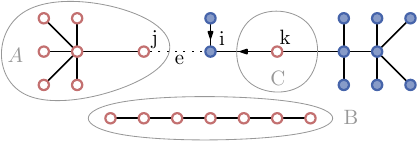}
    \captionof{figure}{Strategy profile $\s$ without edge $e = \{i,j\}$ and $\s'$ where $e$ is bought by $i$. Red (blue) nodes are vulnerable (immunized). $A,B$ and $C$ are vulnerable regions. Without edge $e$ the maximum carnage attack randomizes between infecting a node in $A$ or $B$, while the maximum disruption attacker targets node $k$. In $\s'$, the latter targets a node in~$A$. Arrows indicate edge ownership, directed away from the owner. }
    \label{fig:counter-greed-drop}
\end{figure}
The maximum carnage attacker would treat both vulnerable regions $A$ and $B$ of size $7$ the same, without considering the network structure. In contrast, the maximum disruption attacker would uniquely attack node $k$ in profile~$\s$, since $U_{x^2}(C,\s) = 7^2+2^2+8^2+7^2=166$, while $U_{x^2}(A,\s) = 11^2+7^2=170$. So, targeting region~$C$ minimizes it. However, in $\s'$ we have $U_{x^2}(C,\s') = 9^2+8^2+7^2 = 194$, while $U_{x^2}(A,\s') = 11^2+7^2 = 170$. Hence, $A$ is the only targeted region. 

This shows the counter-intuitive behavior that agent~$i$ prefers to buy the edge $e$ even if it costs $C_E = 9-\varepsilon$, for $\varepsilon>0$, to prevent the destruction of node~$k$. In fact, if, as indicated in \Cref{fig:counter-greed-drop}, agent~$i$ buys no other edges, we have $u_i(\s') = 11 - C_E-C_I$ and $u_i(\s) = 2 - C_I$. 
Such behavior makes the maximum disruption attacker difficult to analyze, as tiny strategy changes may completely shift the targeted regions. Also, a best response strategy may create an edge to a node, like node $j$, that for sure gets destroyed. \hfill $\vartriangleleft$
\end{example}

However, besides the counter-intuitive behavior in~\Cref{example}, the maximum disruption opponent still respects some properties. For instance it is \emph{edge-averse}, i.e., removing edges can only reduce the size of the connected components of each agent, and therefore their utility. We will see that this property is a direct consequence of the strict convexity of the function~$x^2$. Also, this opponent is biased towards targeting the largest components as they contribute more to the social welfare. For our proofs, any function~$f$ that grows faster than $x^2$ will have the same bias. This leads us to naturally encapsulate these properties into a class of attackers, called \emph{super-quadratic-disruptor (\sqda{})}, that we will use for the remainder of this paper to prove statements for maximum disruption opponent and similar attackers.  A {\sqda} is an $f$-opponent such that the function $f$ satisfies: (1) strict convexity, and (2) the function $f(x) / x^2$ is non-decreasing.

\subsection{Related Work}
The study of network formation games has a long tradition in Economics, Computer Science and AI. We restrict our discussion to models and results that are closely related to our work and that of \citet{Goyal16}.  

The objective of ensuring reachability of all other agents was proposed by \citet{bala2000noncooperative} in a setting that is identical to our model but without attack or immunization. They show that equilibrium networks are either empty or trees, depending on the edge price $C_E$. For non-empty equilibria the social welfare is $n^2-\mathcal{O}(n)$, while the empty network has welfare in $\mathcal{O}(n)$. Moreover, they investigate convergence dynamics for finding equilibria. Also, versions with directed edges or connection benefit decay are studied. Later, the authors augmented the model with single edge failures~\cite{BG03} and find that agents build minimally more edges to ensure post failure connectivity. \citet{HS05} extended the model by allowing different failure probabilities per edge and \citet{Kli11,Kli13} studied the price of anarchy of several versions. 

Our reference point, the model by \citet{Goyal16}, is also a direct extension of the reachability model by \citet{bala2000noncooperative}. As one of their main results, they show for the very broad class of well-behaved opponents (and also for multiple stability concepts other than Nash equilibria) that obtained equilibria have at most $2n-4$ edges and that vulnerable regions are trees. Since our $f$-opponents are well-behaved, these results carry over to our analysis. Moreover, the authors show for $C_E,C_I > 1$ that the social welfare of non-trivial equilibria with respect to the maximum carnage and random attack opponent is $n^2 - \mathcal{O}(n^{5/3})$. The same holds for maximum disruption, but only for connected equilibria.
For $C_E>1$, the empty network is an equilibrium with welfare $\mathcal{O}(n)$ and also for $C_I \leq 1$ such bad equilibria exist. Most important for us, they leave the general analysis of the maximum disruption opponent as an open problem. 

The complexity of computing a best response strategy in the model by \citet{Goyal16} was analyzed by \citet{FIKLNS17}. They give a polynomial time algorithm for the maximum carnage and the random attack opponents. Later, \citet{AM23} also achieved this for the maximum disruption attacker. Another extension by \citet{ChenJKKM19} is the natural variant of the model where the infection spreads with probability $p>0$ over each link. For this, the authors only consider the random attack opponent and show that equilibria have at most $\mathcal{O}(n\log n)$ many edges and that equilibria with $\Omega(n)$ edges exist. Also, they show that the expected social welfare of equilibria with $\mathcal{O}(n)$ edges is in $\Theta(n^2)$, i.e., asymptotically optimal and comparable to the setting without attacker. 

Also models with other objectives exist. Close to us is the model by \citet{EchzellLM20}, where agents buy edges to maximize the min-cut in the network. A model with an intermediary, and where edges can have different connection strength was studied by \citet{AnshelevichBK15}. Another large class of related models are network formation games where the objective of an agent is centrality in the created network. Starting from the seminal works of~\citet{jackson1996strategic} and~\citet{Fab03}, many variants have been studied. Among them, there are also versions focusing on the robustness of the created network~\cite{Meirom15,CLMM16}, versions that focus on social networks~\cite{BLLM21} or social distancing~\cite{FGLM22}, variants with temporal edges~\cite{BGKLS23,BGKLS25}, homophilic agents~\cite{BullingerLM22}, or greedy routing~\cite{BFLMR25}.

\subsection{Our Contribution}
We revisit the elegant noncooperative network formation model with attack and immunization by~\citet{Goyal16}. Their main goal was to study the impact of the attacker on the obtained structural properties and the social welfare of equilibrium networks. They find that non-trivial equilibria with respect to the maximum carnage or random attack opponents have a social welfare of $n^2-\mathcal{O}(n^{5/3})$, whereas without attacker the social welfare is $n^2-\mathcal{O}(n)$. Due to the $n^2$-term, this is asymptotically optimal. However, note that the $n^2$-term is a result of using a "benefit minus cost" utility function instead of using a pure cost function. For every agent, every reachable node yields a connection benefit of $1$ whereas unreachable nodes contribute value $0$. This essentially adds the $n^2$-term to the social welfare. In contrast, with a pure cost function, where a reachable node contributes cost $0$ and unreachable nodes yield cost $1$, the $n^2$-term would vanish and the obtained bounds would be far from being tight. We resolve this by proving robust tight bounds on the social welfare. See \Cref{table:results} for an overview.
\begin{table}[h!]
\centering
\begin{tabular}{lcc}
\hline
 & \citet{Goyal16} & Our Results \\
\hline
Max. Carnage& $n^2 - \mathcal{O}(n^{5/3})$ & $n^2 - \mathcal{O}(n)$\\
Random Attack& $n^2 - \mathcal{O}(n^{5/3})$ & $n^2 - \mathcal{O}(n)$ \\
Max. Disruption& open problem &  $n^2 - \mathcal{O}(n)$\\
Tailored & not considered &  $\mathcal{O}(n)$ \\
\hline
\end{tabular}
\caption{Comparison of our results on the social welfare of non-trivial Nash equilibria with the previous results. "Tailored" means a specifically designed opponent.}
\label{table:results}
\end{table}

\noindent
As our main result, we show the optimal social welfare bound of $n^2-\mathcal{O}(n)$ for super-quadratic-disruptor opponents, which subsume the maximum disruption opponent. This shows that even with attackers that try to minimize the post attack social welfare, the obtained welfare is asymptotically the same as in the setting without attacker. Thus, smart strategic agents can cope with strong attackers and still build networks with optimal welfare. This solves an open problem, since for the maximum disruption opponent only a bound of $n^2-\mathcal{O}(n^{5/3})$ for \emph{connected} equilibrium networks was known~\cite{Goyalarxiv}. Thus, we give a general optimal bound for a larger class of attackers.

Finally, we show that the optimal welfare bound does not hold for all attackers. For this, we prove the counter-intuitive result that opponents exist, that enforce asymptotically lower social welfare post attack than the attacker that aims at minimizing this value. In fact, we present a tailored attacker that achieves the low social welfare of empty networks.

\section{Improved Social Welfare for Maximum Carnage and Random Attack}\label{sec:mc_rand}
We study the (expected) social welfare post attack of equilibria with respect to the maximum carnage and the random attack opponent. For this, we focus on the block-cut decomposition of any equilibrium network $G = (V,E)$: 
\begin{definition}
    Consider a network $G=(V,E)$. A node $v \in V$ is called a \emph{cut-vertex} if $G[V\setminus v]$ has more connected components than $G$. A network is \emph{non-separable} if it is connected and has no cut-vertices. Every subnetwork of $G$ that is non-separable and maximal in that property is called a \emph{block}. With that, consider the set of blocks $\mathcal{B}(G)$ and the set of cut-vertices $\mathcal{S}(G)$ of $G$. Connecting all cut-vertices to all blocks they are part of yields a tree $(\mathcal{B}(G)\cup \mathcal{C}(G), E')$ that is called the \emph{block-cut decomposition} of $G$.
\end{definition}

First, we bound the length of specific paths in the block-cut decomposition.
\begin{restatable}{lemma}{lemmaBoundedPathCarnage}\label{lemma:bounded-path-carnage}
    Take the block-cut decomposition $(B(G) \cup C(G), E')$ of an equilibrium network $G=(V,E)$ with respect to the maximum carnage or the random attack opponent. If a path $b_1,c_1,b_2,c_2,...b_r$, that goes through $p$ vulnerable cut-vertices exists, then $p \le 2 C_E + 1$.
\end{restatable}

\begin{proof}[Proof Sketch]
     Let $G(\s)$ be the equilibrium network and consider the vulnerable cut-vertices on the path in the block-cut decomposition of $G(\s)$. 
     Moreover, let $x_1 \in b_1$ and $x_r \in b_r$ be two immunized agents. See \Cref{fig:path-bounded-carnage.pdf} for an illustration.
\begin{figure}[h]
\centering
\includegraphics[width=0.7\linewidth]{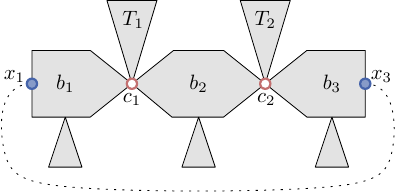}
\caption{The block-cut decomposition of $G(\s)$ with two vulnerable cut-vertices $c_1$ and $c_2$. The dotted edge $\{x_1,x_r\}$ is bought in both strategy profiles $\s_1'$ and $\s_2'$.}
\label{fig:path-bounded-carnage.pdf}
\end{figure}

\noindent
    Now consider the strategy profiles $\s_1'$ and $\s_2'$ that result from either agent $x_1$ buying an edge to agent $x_r$ or vice versa in strategy profile $\s$. Note, that the connectivity of agent~$x_1$ in profile~$\s_1'$ and of agent~$x_r$ in profile~$\s_2'$ is the same. 

    By computing the obtained utilities, we show that the desired result is a consequence of the fact that neither agent~$x_1$ nor agent~$x_r$ can improve from this deviation. 
\end{proof}
\noindent Next, we use the concept of a centroid of a tree, extend it to connected components and show that it always exists.

\begin{restatable}{definition}{defCentroid}\label{def:centroid}
    Let $Z$ be a connected component in $G(\s)$. A node $c \in Z$ is called \emph{centroid} of Z if it respects the property that if it immunizes (or stays immunized if it already was), then after any vulnerable region in $Z$ is destroyed, the size of the connected component containing $c$ is at least $|Z|/2$.
\end{restatable}
\begin{restatable}{lemma}{lemmaCentroid}
Every connected component of $G$ has a centroid.\label{lemma:centroid}
\end{restatable}

\noindent From the proof of \Cref{lemma:centroid}, we get the following:
\begin{restatable}{corollary}{lemmaNoCentroidImpliesVulnerable}\label{lemma:no-centroid-implies-vulnerable}
    In any connected component $Z$ of a network $G(\s)$, one of the two following properties must hold:
    \begin{itemize}
        \item there exists an immunized centroid of $Z$,
        \item there exists a vulnerable region containing a centroid that if removed, every remaining connected component has size $< |Z|/2$.
    \end{itemize}
\end{restatable}

\noindent In the following, we use a result from~\citet{Goyal16} that holds for all well-behaved opponents.
\begin{restatable}{lemma}{lemmaEdgeNumber}{(Lemma 2 \& Theorem 2 of \citet{Goyal16})}\label{lemma:1}
    Let $G=(V,E)$ be an equilibrium network with a well-behaved opponent. Then
    \begin{enumerate}
        \item all vulnerable regions in $G$ are trees, and
        \item $|E| \leq 2n-4$, for all $n \geq 4$.
    \end{enumerate}
\end{restatable}

\noindent Now, we can prove the main result of this section.
\begin{restatable}{theorem}{theoremBoundsPrevious}\label{theorem:bounds-previous}
    The social welfare of non-trivial Nash equilibria with respect to the maximum carnage or the random attack opponent is $n^2 - O(n)$.
\end{restatable}

\begin{proof}[Proof Sketch]
    Consider the immunized agent~$x$ and the block-cut decomposition of the network rooted in the block containing node~$x$. 
    For every $0 \le i \le 2C_E +1$, define $H_i$ the set of vulnerable cut-vertices of rank $i$. The rank is defined as follows. Let $H_0$ is the set of vulnerable cut-vertices such that no other vulnerable cut-vertices are in the subtree of this node. Then $H_{i+1}$ is the set of cut-vertices whose subtree contains a different node in $H_i$. Finally, let $H = \cup_{i} H_i$ the set of all vulnerable cut-vertices. See \Cref{fig:bounds-carnage}.
    \begin{figure}[h]
    \centering    \includegraphics[width=0.5\linewidth]{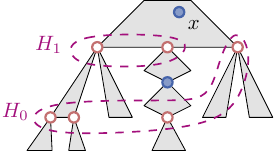}
    \caption{The layers $H_i$ from the proof of \Cref{theorem:bounds-previous}.}
    \label{fig:bounds-carnage}
    \end{figure}

    \noindent
    Note that for a given $0 \le i \le 2 C_E +1$ and a node $y$, there is at most one node in $H_i$ that disconnects nodes~$x$ and $y$ when removed. Therefore,
        $\sum_{z \in H_i} CC_{x}(z, \s) \ge n (|H_i|-1)$,
    since every node is in all of these components except for at most one. Also, all targeted regions are singletons. Therefore, 
    $$\E_{\U(\s)}[CC_x(\s)] \ge\; n - 2 - 4 C_I( C_E +1).$$

    Finally the social welfare is the expected value of the sum of squares of connected components minus the cost. This is higher or equal to the expected value of the square of the connected component of $x$ minus the cost. We therefore have, using Jensen's inequality and \Cref{lemma:1}, that the social welfare of Nash equilibrium~$\s$ is at least
    \begin{align*}
    & \E_{\U(\s)}[CC_x(\s)^2] - (2n-4)C_E - n C_I  \\
    \ge\;&  \E_{\U(\s)}[CC_x(\s)]^2 - (2n-4)C_E - n C_I \\
    \ge\;& (n - 2 - 4 C_I(C_E + 1))^2  - (2n-4)C_E - n C_I \\
    =\;& n^2 - O(n). \qedhere
    \end{align*}
\end{proof}

\section{Optimal Welfare Bounds for \sqda{} Opponents}
In this section, all our statements are with respect to a fixed game $(n,C_E,C_I,\mathcal{A})$ with $C_E, C_I > 1$, where $\mathcal{A}$ is an \sqda{} opponent with function $f$. Since $\mathcal{A}$ is fixed for this section, we will omit it as a parameter and in subscripts.

\noindent
We show the intuitive statement, that the \sqda{} opponent favors splitting components into small components. Also, selling edges in targeted regions makes the region less attractive for the \sqda{} opponent, i.e., it is edge-averse. Both follows from \Cref{lemma:f-opp-increase} (see Supp. Material) which states that \sqda{} opponents favor attacking large components.
\begin{restatable}{corollary}{sqdCutVertices}\label{cor:sqd-cut-vertices}
    Let $\s$ be a strategy profile such that a component $K$ in $G(\s)$ contains at least two vulnerable regions $T_\text{cut}, T_\text{leaf}$ of equal size, such that the remaining nodes in component $K$ stay connected when removing $T_\text{leaf}$ but not when removing $T_\text{cut}$. Then $U_f(T_\text{cut}, \s) < U_f(T_\text{leaf}, \s)$, which means that $T_\text{leaf}$ is not targeted by the \sqda{} opponent.
\end{restatable}

\begin{restatable}{corollary}{corTechCor}\label{cor:tech-cor2}
    Let $\s$ be a strategy profile and consider a vulnerable node~$u$ inside a vulnerable region $T$ in $G(\s)$. Further, let $\s'$ be the strategy profile derived from $\s$, where
    \begin{itemize}
        \item some edges are sold that are incident to nodes in $T$, and/or
        \item some agents in $T$ got immunized.
    \end{itemize}
    Then, if $T'$ is the vulnerable region containing node~$u$ in network $G(\s')$, it holds that $U_f(T', \s') \geq U_f(T,\s)$.
\end{restatable}
\noindent

The next corollary is one of the main tools for our other results. It follows from \Cref{lemma:greed-drop} (see Supp. Material).

\begin{restatable}{corollary}{corSizeAfterEdge}\label{cor:size-after-edge}
    Let $u$ be a cut-vertex in an equilibrium network w.r.t. an \sqda{} opponent $\mathcal{A}$, and $Z$ be a remaining connected component after node~$u$ is removed. If agent~$u$ bought all of the $k$ edges that it shares with nodes in $Z$, and if there are either targeted regions outside of $Z$ or no targeted regions included in $Z$, then it holds that $|Z| \ge kC_E$. Also, if node~$u$ was targeted, then $|Z| > kC_E$.
\end{restatable}
\noindent 
Next, we reprove and generalize a result from the full-version of \cite{Goyal16} (see \cite{Goyalarxiv}). 

\begin{restatable}{lemma}{lemmaTargetsSingleton}{(generalizes Theorem 6 in \cite{Goyalarxiv})}\label{lemma:targetsingletons}
    In any connected component with at least one immunized node, the targeted regions are singletons for equilibrium networks w.r.t. an SQD opponent.
\end{restatable}

\noindent
The next results consider what happens if the \sqda{} opponent targets cut-vertices of the network.  
\begin{restatable}{lemma}{lemmaConditionCutTargeted}\label{lemma:condition-cut-targeted}
    In a Nash equilibrium with respect to an $SQD$ opponent~$\mathcal{A}$, if a connected component contains an immunized node and a targeted cut-vertex, then, 
    \begin{enumerate}[label=\roman*),topsep=0pt,itemsep=2pt,leftmargin=20pt]
        \item when removing any targeted node, the size of all remaining connected components is strictly higher than $C_E$, and
        \item there are at least two targeted regions outside of this component.
    \end{enumerate}
\end{restatable}

\begin{restatable}{corollary}{lemmaNoIsolatedCut}\label{lemma:no-isolated-cut}
    In a Nash equilibrium $\s$ with respect to an $SQD$ opponent $\mathcal{A}$, if a connected component contains an immunized node and a targeted cut-vertex, then there is no isolated node in $G(\s)$.
\end{restatable}

\noindent
Next, we bound the sizes of connected components.

\begin{restatable}{lemma}{lemmaMinimalSize}\label{lemma:minimal_size}
    In a Nash equilibrium $\s$ with respect to an $SQD$ opponent $\mathcal{A}$, the size of a connected component is either $1$ or at least $C_E+1$.
\end{restatable}

\noindent 
The next lemma deals with the implications of the existence of a targeted region that is a cut-vertex inside a component with an immunized node.

\begin{restatable}{lemma}{singleComponentCut}\label{single-component-cut}
    In a Nash equilibrium $\s$ with respect to an \sqda{} opponent $\mathcal{A}$, if two different connected components each contain an immunized node and a targeted region, then none of the targeted regions are cut-vertices.
\end{restatable}
\noindent 
Later will be bound the number of agents in some cases of Nash equilibria. The following statement is useful for this.

\begin{restatable}{lemma}{lemmaBoundedNumberComponentsCut}\label{lemma:bounded-number-components-cut}
    In a Nash equilibrium $\s$ with respect to an \sqda{} opponent $\mathcal{A}$, the number of connected components is bounded by $C_E + C_I +2$ if there are no isolated nodes.
\end{restatable}

\noindent
The next lemma allows us to eliminate the possibility of targeted cut-vertices as they only exist in small networks and thus do not impact the asymptotic bounds.

\begin{restatable}{lemma}{lemmaCutImplyBounded}\label{lemma:cut-imply-bounded}
    If there is a Nash equilibrium $\s$ with respect to an \sqda{} opponent $\mathcal{A}$ such that $G(\s)$ has a component with a targeted cut-vertex and an immunized node, then the number of agents in this game instance is at most
        $$(C_I+C_E+2)(2C_I + 3C_E).$$
\end{restatable}

\noindent 
Now we consider a general result that guarantees the existence of a specific connected component.

\begin{restatable}{lemma}{lemmaTargetedWithImmunized}\label{lemma:targeted-with-immunized}
    In a non-trivial Nash equilibrium w.r.t. an \sqda{} opponent, if there is a targeted region, at least one such region is in a connected component with an immunized node.
\end{restatable}

\noindent 
We continue with the setting without targeted cut-vertices and show the existence of at least two vulnerable nodes.

\begin{restatable}{lemma}{lemmaTwoVulnerable}\label{lemma:two-vulnerable}
    In a non-trivial Nash equilibrium $\s$ with respect to an \sqda{} opponent, if there are no targeted cut-vertices in connected components with an immunized node, every connected component with an immunized node and a targeted region contains at least two vulnerable nodes.
\end{restatable}
\noindent 
Finally, we can show that most one connected component with a targeted regions exists.

\begin{restatable}{lemma}{lemmaSoloImmunizedTargeted}\label{lemma:solo-immunized-targeted}
    In a non-trivial Nash equilibrium with respect to an \sqda{} opponent, if there are no targeted cut-vertices in connected components with an immunized node, then there is at most one connected component with an immunized agent and a targeted region.
\end{restatable}

Now, we derive the important statement, that without targeted cut-vertices no isolated agents exist. This is helpful, since such agents contribute to low social welfare. 

\begin{restatable}{lemma}{lemmaNoIsolated}\label{lemma:no-isolated}
   In a non-trivial Nash equilibrium w.r.t. an \sqda{} opponent, if no targeted cut-vertices exist in connected components with an immunized node, no agent is isolated. 
\end{restatable}

\begin{proof}[Proof Sketch]
    Let $\s$ be a non-trivial Nash equilibrium. By \Cref{lemma:targeted-with-immunized}, there is a component $K$ with at an immunized agent and, if there are vulnerable nodes in $G(\s)$, a targeted region. By \Cref{lemma:targetsingletons}, targeted regions must be singletons and by assumption there are no targeted cut-vertices. Hence, every vulnerable agent in $K$ is targeted and not a cut-vertex.
    
    Now, towards a contradiction, assume there exists an isolated agent~$x$. If agent~$x$ was already immunized, then buying an edge to any immunized node in $K$ would be a strict utility increase, as there are at least $C_E+1$ nodes in $K$ and no vulnerable cut-vertex exists in $K$. Therefore, agent~$x$ is vulnerable, which means that there exists a targeted region, and thus targeted non-cut-vertices in component~$K$.
    
    Now there will be several cases depending on the structure of component~$K$, see \Cref{fig:no-isolated-main} for an illustration. 
    \begin{enumerate}
        \item Some agent in $K$ bought two edges.
        \item There is only one immunized agent in $K$.
        \item An immunized agent bought an edge that disconnects $K$ when removed and no agent bought more than one edge.
        \item Every edge bought by an immunized agent in $K$ is part of a cycle, no agent bought more than one edge, and some immunized agent bought an edge.
    \end{enumerate}
    \begin{figure}[h]
        \centering
        \includegraphics[width=\linewidth]{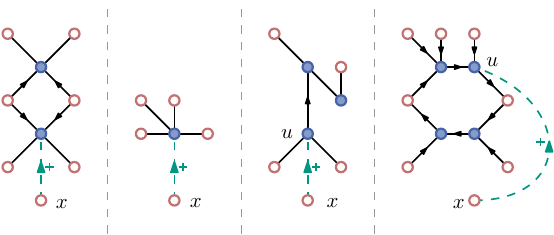}
        \caption{The four cases of the proof. The dashed edge is the proposed deviation.}
        \label{fig:no-isolated-main}
    \end{figure}
      
    \noindent
    We find that in all cases buying an edge to an immunized node of $K$ will be a strictly improving deviation for agent $x$. Without buying immunity in cases $1,2$ and $4$ and sometimes buying immunity in Case~$3$.
\end{proof}
\noindent 
Finally, we combine all our structural observations, to show that large enough Nash equilibria must be connected.

\begin{restatable}{lemma}{lemmaBoundedDisconnected}\label{lemma:bounded-disconnected}
    All non-trivial Nash equilibria with respect to an \sqda{} opponent that are not connected have a bounded number of agents.
\end{restatable}

\begin{proof}    
    First of all, \Cref{lemma:cut-imply-bounded} ensures that if the number of nodes is greater than $(C_E + C_I +2)(2C_I +3C_E)$, then there are no targeted cut-vertices in connected components with immunized nodes. Assume this is the case, i.e., that we have more than  $(C_E + C_I +2)(2C_I +3C_E)$ many agents. 
    
    Now if there are no vulnerable agents, then non-trivial Nash equilibria are connected. This was already proven by \citet{bala2000noncooperative}, since without vulnerable agents we are in the setting without attacker. 
    
    Hence, we assume that at least one vulnerable agent exists. For this, by \Cref{lemma:targeted-with-immunized}, we know that there is a  connected component $K$ with an immunized agent and a targeted region. There are also no isolated nodes, according to \Cref{lemma:no-isolated}, and therefore, only a bounded number of connected components, as guaranteed by \Cref{lemma:bounded-number-components-cut}. 
    
    By \Cref{lemma:targetsingletons}, component~$K$ has targeted singletons that are not cut-vertices. Now we show, that $K$ is the largest component with vulnerable regions. Indeed, assume another component $M$ is strictly larger than $K$. If $M$ contains vulnerable nodes, then with the same reasoning as above, the targeted regions in $M$ are singletons and not cut-vertices. Thus, if $|M|>|K|$ the \sqda{} opponent would attack only regions in $M$, which contradicts that component $K$ has targeted regions. If, on the other hand, component $M$ only contains immunized nodes, then there is a strictly improving deviation for the immunized agents of $K$. This consists of buying an edge to any node in $M$. Since, by \Cref{lemma:no-isolated}, component $M$ contains at least two nodes, we get, by \Cref{lemma:minimal_size}, that $|M| \geq C_E+1$. Thus, buying the edge is profitable.
    
    Therefore, we can assume that $K$ is the largest connected component. Let $L$ be the second largest component. We now find an agent in $L$ and consider a specific strategy change, which then yields a bound on the size of $K$. For this, take a centroid~$c$ of component~$L$, and consider the deviation of agent~$c$ immunizing (if it was not already immunized) and buying an edge to an immunized node of component~$K$. The expected size of the connected component of agent~$c$ after the deviation is at least 
        $|K| + \frac{|L|}{2} \ge |L| + \frac{|K|}{2}$.
    This holds, since if a vulnerable region in $K$ is attacked, then agent~$c$'s connected component has size $|K|-1 + |L| \geq |K| + \frac{|L|}{2}$, since $|L|\geq 2$. If a vulnerable region in $L$ is attacked, then $K$ is not affected and, since agent~$c$ is an immunized centroid of component $L$, the size of its connected component is at least $|K| + \frac{|L|}{2}$. If a vulnerable region outside of $K$ and $L$ is attacked, then the size of agent~$u$'s connected component is $|K|+|L| > |K| + \frac{|L|}{2}$.
    
    Hence, for the original strategy profile to be a Nash equilibrium, the inequality $|K| \le 2(C_I+C_E)$ must hold, since the additional cost of the deviation of agent~$c$ is $C_I+C_E$, while the expected additional connectivity is at least $\frac{|K|}{2}$.

    Hence, the number of agents is bounded by $(C_E + C_I +2)(2C_I +2C_E)$,
    if there are no targeted cut-vertices. By \Cref{lemma:cut-imply-bounded}, the latter holds. The stated upper bound on the number of agents is true, since by \Cref{lemma:bounded-number-components-cut} and \Cref{lemma:no-isolated}, we have $C_E + C_I +2$ many connected components, each of size at most $|K|\leq 2C_I +2C_E$.

    This bound of  $(C_E + C_I +2)(2C_I +2C_E)$ is lower than our assumed minimum number of agents of $(C_E + C_I +2)(2C_I +3C_E)$, which is a contradiction. 
    Thus, every non-trivial Nash equilibrium with more than $(C_E + C_I +2)(2C_I +3C_E)$ nodes must be connected.
\end{proof}

Now we are ready, to prove our main result of this paper.
\begin{restatable}{theorem}{theoremBoundsDisruption}\label{theorem:bounds-disruption}
    The social welfare of any non-trivial Nash equilibrium with respect to an \sqda{} opponent is $n^2 - \mathcal{O}(n)$.
\end{restatable}

\begin{proof}
    Because this is an asymptotic result, we can ignore the social welfare of instances with bounded size. Thus, for a growing number of agents $n$,  \Cref{lemma:bounded-disconnected} ensures that every non-trivial Nash equilibrium is connected. With \Cref{lemma:targetsingletons}, we know that every targeted region is a singleton. Finally, \Cref{lemma:cut-imply-bounded} ensures that no targeted regions are cut-vertices.

    Therefore, if some agents are vulnerable, the size of the connected component after the attack is $n-1$. And because the total cost for all edges and immunizations are linear in $n$, as seen in \Cref{lemma:1}, the social welfare is $n^2 - \mathcal{O}(n)$.
\end{proof}

\section{Low Welfare for a Tailored Opponent}
As contrast for our positive results on the maximum carnage, random attack, and the \sqda{} opponents, we show that there exist opponents, such that Nash equilibria can have social welfare of $\Theta(n)$, i.e., the lowest possible welfare. This shows the counter-intuitive result, that opponents exist, that achieve a lower social welfare than the opponent that actually aims for minimizing the social welfare.

\begin{restatable}{theorem}{thmBadOpponent}\label{thm:bad-opponent}
    There are non-trivial Nash equilibria with respect to some attacker $\mathcal{A}$ with a social welfare of $\Theta(n)$.
\end{restatable}

\begin{proof}[Proof Sketch]
We consider a family of instances with $n$ agents and $C_E= C_I = 6$ and a tailored $f$-opponent $\mathcal{A}$, where $f$ is defined as follows: 

\begin{center}
    \begin{tabular}{c|c|c|c|c|c|c}
        $i$ & $1$ & $7$ & $8$ & $9$ & $10$ & otherwise \\ \hline
        $f(i)$ & $2$ & $3$ & $5$ & $4$ & $7$ & $0$ 
    \end{tabular}
  \end{center} 

For every $n \geq 10$, we consider the strategy profile~$\s_{\text{bad}}$ illustrated in \Cref{fig:bad-Nash}. There, all but nine agents are isolated.
\begin{figure}[h]
    \centering
    \includegraphics[width=0.6\linewidth]{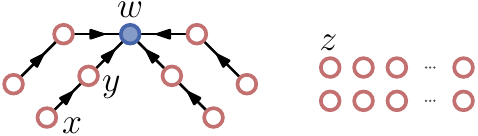}
        \captionof{figure}{The Nash equilibrium~$\s_{\text{bad}}$ with social welfare in $\Theta(n)$, for our tailored opponent, with $C_E = C_I = 6$.}
        \label{fig:bad-Nash}
\end{figure}
 There are therefore $n-5$ vulnerable regions, and $n-9$ targeted nodes. 
 Note, that one of the isolated nodes will be targeted. Thus, the social welfare of profile~$\sbad$ is $$(n-9)\cdot \frac{n-10}{n-9}\cdot 1 + 9\cdot 9 - 8 C_E - C_I \in \Theta(n),$$ since each of the $n-9$ isolated nodes survives with probability $\frac{n-10}{n-9}$ and since the nine nodes in the large connected component are not targeted, each of them has connectivity~$9$.   

We verify that profile $\sbad$ is a Nash equilibrium. 
\end{proof}

\section{Conclusion}
We have revisited the elegant noncooperative network formation model with attack and immunization from \citet{Goyal16}. We prove optimal social welfare of Nash equilibria under various strong types of attackers. This solves an open problem and highlights, that decentralized strategic agents can create very robust networks even under attack.

Our last result, the tailored opponent that yields Nash equilibria with low social welfare, indicates that there is still much to explore even for this very basic model. For example, characterizing what kind of attackers yield suboptimal social welfare, e.g., which properties of the $f$-function are crucial. Also, other variants of the model, for example with cooperation among the agents, could be studied.

\bibliography{references}

\clearpage

\appendix

\section{Omitted Details from the Analysis of Maximum Carnage and Random Attack}

\lemmaBoundedPathCarnage*

\begin{proof}
     Let $G(\s)$ be the equilibrium network and let $c_{i_1},...,c_{i_p}$ be the extraction of the vulnerable cut-vertices on the path in the block-cut decomposition of $G(\s)$. Let $T_{i_j}$ denote the subtree of cut-vertex $c_{i_j}$, i.e., the part of the block-cut decomposition that gets disconnected from blocks $b_{i_j}$ and $b_{i_j+1}$ when $c_{i_j}$ is deleted. Moreover, let $|T_{i_j}|$ denote the number of nodes in $T_{i_j}$ and note that $c_{i_j}$ belongs to $T_{i_j}$. Moreover, let $x_1 \in b_1$ and $x_r \in b_r$ be two immunized agents. See \Cref{fig:path-bounded-carnage.pdf_appendix} for an illustration.
\begin{figure}[h]
\centering
\includegraphics[width=0.7\linewidth]{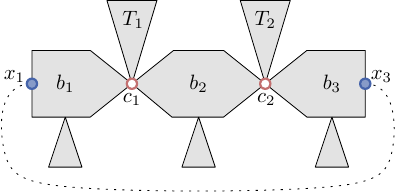}
\caption{The block-cut decomposition of $G(\s)$ with two vulnerable cut-vertices $c_1$ and $c_2$. The dotted edge $\{x_1,x_r\}$ is bought in both strategy profiles $\s_1'$ and $\s_2'$.}
\label{fig:path-bounded-carnage.pdf_appendix}
\end{figure}

\noindent
    Now consider the strategy profiles $\s_1'$ and $\s_2'$ that result from either agent $x_1$ buying an edge to agent $x_r$ or vice versa in strategy profile $\s$. Note, that the connectivity of agent~$x_1$ in profile~$\s_1'$ and of agent~$x_r$ in profile~$\s_2'$ is the same. 

    To compute the utility difference of both strategy changes, note that every vulnerable cut-vertex on the path has a probability of at least $\tfrac{1}{n}$ of being attacked, and such an attack disconnects the components containing the agents $x_1$ and~$x_r$. Also, any node~$z$ can be outside of both resulting connected components only if it is in the sub-tree of the targeted cut-vertex, which can only hold for one of them.
    
    This justifies the following inequality for the expected size of the connected components for both agents $x_1$ and $x_r$ before and after their strategy change:

    \begin{align*}
     &\E_{\T(\s)}[CC_{x_1}(\s_1')] - \E_{\T(\s)}[CC_{x_1}(\s)]\\
     &\quad\quad+ \E_{\T(\s)}[CC_{x_r}(\s_2')] - \E_{\T(\s)}[CC_{x_r}(\s)]\\
     \geq & \frac{1}{n}\sum_{j=1}^{p} \big(CC_{x_r}(c_{i_j},\s) + CC_{x_1}(c_{i_j},\s)\big)\\
     = & \frac{1}{n} \sum_{j=1}^p \big(n - |T_{i_j}|\big) = \frac{1}{n}\left(pn - \sum_{j=1}^p |T_{i_j}|\right)\\
     \geq & (p-1),
    \end{align*}
    where the last inequality follows, since $\sum_{j=1}^p |T_{i_j}| < n$.
    
    Therefore, at least one of the two values $\E_{\T(\s)}[CC_{x_1}(\s_1')] - \E_{\T(\s)}[CC_{x_1}(\s)]$ or $\E_{\T(\s)}[CC_{x_r}(\s_2')] - \E_{\T(\s)}[CC_{x_r}(\s)]$ must be at least $\frac{p-1}{2}$. Since $G(\s)$ is an equilibrium network, none of the strategy changes can improve the utility of the respective agent. This implies that $C_E \ge \frac{p-1}{2}$ holds.
\end{proof}

\lemmaCentroid*
\begin{proof}
Let $Z$ be any connected component of a network $G$. To construct a centroid, start with spanning trees of every vulnerable region of~$Z$, then arbitrarily extend this forest to a spanning tree~$L$ of~$Z$.
Let $c$ be a centroid of $L$. By definition, when $L$ is rooted in $c$, every proper subtree contains at most $|Z|/2$ nodes. Now, we show that $c$ fulfills \Cref{def:centroid}.

If agent~$c$ is immunized or buys immunity, removing a vulnerable region does not remove $c$. However, because the spanning tree $L$ was from the forest of spanning trees of vulnerable regions, the attacked vulnerable region is also a vulnerable region in the tree $L$. Because $L$ is connected, the attacked region can be at most a complete subtree of $c$, without $c$ itself. Hence, agent~$c$'s connected component of after removing the attacked region has size at least $|Z|/2$.
\end{proof}

\lemmaNoCentroidImpliesVulnerable*
\begin{proof}
    Consider a centroid $c$ of $Z$. If agent~$c$ is immunized, the first property holds and we are done. Thus, we assume that agent~$c$ is vulnerable and we investigate the remaining connected components after removing the vulnerable region~$T$ containing agent~$c$. If all of these components have size strictly less than $|Z|/2$, then the second property holds. 

    Thus, assume that component~$Z_x$ has size at least $|Z|/2$. Take any node~$x \in Z_x$ that shared an edge with the removed region~$T$. We claim that $x$ is an immunized centroid.

    First of all, because node~$x$ was not destroyed, it must be immunized. It remains to show that node~$x$ is a centroid of~$Z$. We will inherit the desired property from the centroid~$c$.
    Since $c$ is a centroid, we have that after removing a vulnerable region that does not contain node~$c$, the size of the connected component of $c$ is greater or equal to $|Z|/2$. As the connected component containing node~$x$ also contains node~$c$, it has the desired size. Now, if the removed region is $T$, then the connected component of node~$x$ contains also component~$Z_x$ and thus has size at least $|Z|/2$.
\end{proof}

Also, we use two properties proven by~\citet{Goyal16} for the maximum carnage and random attack opponents.

\begin{restatable}{lemma}{lemmaSingletons}{(Theorem 3 in \cite{Goyal16})}\label{lemma:maxrand_connected}
For non-trivial equilibria with $C_E>1$ and the maximum carnage or random attack opponent, all Nash equilibria are connected. 
\end{restatable}
 
\begin{restatable}{lemma}{lemmaSingletons}{(Lemma 3 in \cite{Goyal16})}\label{lemma:maxrand_singletons}
For non-trivial equilibria with $C_E>1$ and the maximum carnage or random attack opponent, for any connected component that has at least one immunized agent and at least one edge, the targeted regions are singletons. 
\end{restatable}
\noindent Note that when the conditions of \Cref{lemma:maxrand_singletons} hold, then we have that $\E_{\U(\s)}[CC_v(\s)] \coloneqq \E_{\T(\s)}[CC_v(\s)]$, since all targeted regions are singletons.

Finally, we are ready bound the number of vulnerable agents in any non-trivial Nash equilibrium.

\begin{restatable}{lemma}{lemmaLowerBoundVulnerable}\label{lemma:lower-bound-vulnerable}
    In a non-trivial Nash equilibrium $\s$ with respect to the maximum carnage or random attack opponent, there are either no vulnerable agents or at least $\frac{n}{2C_I}$ of them.
\end{restatable}

\begin{proof}
    If all agents are immunized, then clearly the number of vulnerable agents is zero. Hence, 
    assume that there are vulnerable agents in $G(\s)$, i.e., $\mathcal{U}(\s) \neq \emptyset$, and let node~$x$ be a centroid of the network $G(\s)$.

    Since profile~$\s$ is a non-trivial Nash equilibrium, we know that at least one agent is immunized. Also, since $C_E>1$, we have that \Cref{lemma:maxrand_connected} that $G(\s)$ is connected and, by \Cref{lemma:maxrand_singletons}, we know that all targeted regions are singletons. 
    
    If agent~$x$ is immunized, let $v$ denote a vulnerable node that is connected via an edge with the immunized region containing the centroid~$x$.  If $x$ is vulnerable, then let~$v = x$.

   Consider the strategy profile~$\s'$ that is identical to profile~$\s$ but where agent~$v$ additionally buys immunity. Agent~$v$'s expected size of its connected component in $G(\s)$ is $$\E_{\U(\s)}[CC_v(\s)] = \frac{|\U(\s)|-1}{|\U(\s)|}\E_{\U(\s) \setminus \{v\}}[CC_v(\s)].$$ Because agent~$v$'s connected component is the same as the one of agent~$x$, which is a centroid, when node~$v$ is not destroyed, it holds that $\E_{\U(\s) \setminus \{v\}}[CC_v(\s)] \ge \frac{n}{2}$, since, by \Cref{lemma:maxrand_connected}, the network $G(\s)$ is connected. Therefore, in strategy profile~$\s'$ with vulnerable nodes $\U(\s') = \U(\s)\setminus \{v\}$ we have
   \begin{align*}
    \E_{\U(\s')}[CC_v(\s')] &\geq \E_{\U(\s) \setminus \{v\}}[CC_v(\s)]\\ &\ge \E_{\U(\s)}[CC_v(\s)] + \frac{n}{2|\U(\s)|}.
    \end{align*}
    Because strategy profile~$\s$ is a Nash equilibrium, we get that $\frac{n}{2|\U(\s)|} \le C_I$ and therefore that $|\U(\s)| \ge \frac{n}{2 C_I}$.
\end{proof}

\theoremBoundsPrevious*
\begin{proof}
    Consider the immunized agent~$x$ and the block-cut decomposition of the network rooted in the block containing node~$x$.

    For every $0 \le i \le 2C_E +1$, define $H_i$ the set of vulnerable cut-vertices of rank $i$. The rank is defined as follows. Let $H_0$ is the set of vulnerable cut-vertices such that no other vulnerable cut-vertices are in the subtree of this node. Then $H_{i+1}$ is the set of cut-vertices whose subtree contains a different node in $H_i$. Finally, let $H = \cup_{i} H_i$ the set of all vulnerable cut-vertices. See \Cref{fig:bounds-carnage_appendix} for an illustration.
    \begin{figure}[h]
    \centering    \includegraphics[width=0.6\linewidth]{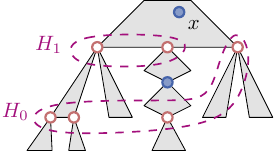}
    \caption{The layers $H_i$ from the proof of \Cref{theorem:bounds-previous}.}
    \label{fig:bounds-carnage_appendix}
    \end{figure}

    \noindent
    Note that for a given $0 \le i \le 2 C_E +1$ and a node $y$, there is at most one node in $H_i$ that disconnects nodes~$x$ and $y$ when removed. Therefore,
    \begin{equation}\label{eq1:bounds-previous}
        \sum_{z \in H_i} CC_{x}(z, \s) \ge n (|H_i|-1),
    \end{equation}
    since every node is in all of these components except for at most one. Also, remember, that by \Cref{lemma:maxrand_singletons}, all targeted regions are singletons. Therefore, 
    \begin{align*}
    &\E_{\U(\s)}[CC_x(\s)] \\
    =\;& \frac{|\U(\s)| - |H|}{|\U(\s)|}(n-1) + \frac{1}{|\U(\s)|}\sum_{i=0}^{2C_E + 1} \sum_{z \in H_i} CC_x(\s,z) \\
    \overset{\text{Eq.} (\ref{eq1:bounds-previous})}{\ge}\;& \frac{|\U(\s)| - |H|}{|\U(\s)|}(n-1) + \frac{1}{|\U(\s)|}\sum_{i=0}^{2 C_E + 1} n(|H_i| - 1) \\
    =\;& \frac{|\U(\s)| - |H|}{|\U(\s)|}(n-1) + \frac{|H|}{|\U(\s)|}n - (2 C_E + 2) \frac{n}{|\U(\s)|}\\
    \ge\;& n -1 - \frac{|H|}{|\U(\s)|} - 4 C_I( C_E +1) \\
    \ge\;& n - 2 - 4 C_I( C_E +1),
    \end{align*}
    where for second to last inequality, we used $\U(\s) \geq\frac{n}{2C_I}$ from \Cref{lemma:lower-bound-vulnerable}. 
    
    Finally the social welfare is the expected value of the sum of squares of connected components minus the cost. This is higher or equal to the expected value of the square of the connected component of $x$ minus the cost. We therefore have, using Jensen's inequality and \Cref{lemma:1}, that the social welfare of Nash equilibrium~$\s$ is at least
    \begin{align*}
    & \E_{\U(\s)}[CC_x(\s)^2] - (2n-4)C_E - n C_I  \\
    \ge\;&  \E_{\U(\s)}[CC_x(\s)]^2 - (2n-4)C_E - n C_I \\
    \ge\;& (n - 2 - 4 C_I(C_E + 1))^2  - (2n-4)C_E - n C_I \\
    =\;& n^2 - O(n). \qedhere
    \end{align*}
\end{proof}

\section{Omitted Details from the Analysis for SQD Opponents}
The following lemma captures that the \sqda{} opponent favors attacking large components.

\begin{restatable}{lemma}{lemmaFOppIncrease}\label{lemma:f-opp-increase}
    Let $\s$ and $\s'$ be two strategy profiles and let $T$ and $T'$ be vulnerable regions with respect to $\s$ and $\s'$, respectively. Further, let $(CC_i(T, \s))_i$ be the vector of the sizes of the connected component every agent~$i$ lies in with respect to $\s$ after an attack on $T$, and $\varphi$ a permutation of the agents $[n]$. Then, if $(CC_i(T,\s))_i \ge (CC_{\varphi(i)}(T',\s'))_i$, then $U_f(T, \s) 
\geq U_f(T',\s')$. And if additionally $CC_i(T,\s) > CC_{\varphi(i)}(T',\s')\ge1$ for some $i \in [n]$, then $U_f(T, \s) > U_f(T',\s')$.
\end{restatable}
\begin{proof}
    Remember, the function $U_f(T,s)$ is defined to be $U_f(T,\s) \coloneqq \sum_{K \in \mathcal{K}(T)}f(|K|)$, where $\mathcal{K}(T)$ is the set of connected components of the network after $\mathcal{A}$ attacked some vulnerable region $T \in \mathcal{T}(\s)$. We have $$U_f(T,\s) = \sum_{K \in \mathcal{K}(T)}f(|K|) = \sum_{v\in V} \frac{f(CC_v(T, \s))}{CC_v(T,s)},$$ where the denominator compensates the fact that every connected component occurs exactly as often as it has vertices in this sum. By convention, if $CC_v(T,\s) = 0$ then $ \frac{f(CC_v(T, \s))}{CC_v(T,s)} = 0$.
    
    Now, in order to prove this lemma, it is sufficient to prove that $n \mapsto \frac{f(n)}{n}$ is strictly increasing, as the order of the summation does not matter (which is why we can consider any permutation~$\varphi$ of the agents). We do so, by showing that \begin{equation}f(n) < n [f(n+1) - f(n)] \label{f_eq}\end{equation} holds for every $n > 0$ (which is a rephrased version of $\frac{f(n)}{n} < \frac{f(n+1)}{n+1}$ for all $n > 0$). For this, notice that $f(n)$ can be written as the telescoping series 
        $$f(n) = f(0) + \sum_{i=0}^{n-1} \left(f(i+1) - f(i)\right).$$
    We know that, by definition, $f(0)=0$ holds, and $$f(i+1) - f(i) < f(n+1) - f(n)$$ (for all $i < n$) by the strict-convexity of $f$. With this, it follows that $(\ref{f_eq})$  holds and thus also the lemma.
\end{proof}

\sqdCutVertices*
\begin{proof}
    When removing $T_\text{cut}$, the remaining nodes of component~$K$ form several connected components $K_1,...,K_r$. This implies that their respective sizes are all strictly lower than $|K| - |T_\text{cut}|$. However, when removing $T_\text{leaf}$, the connectivity of the remaining nodes is exactly $|K| - |T_\text{leaf}|$.
    
    We have that, (1) there is the same number of remaining nodes in both cases, since we assumed that $|T_\text{leaf}| = |T_\text{cut}|$, and (2) that the size of the other components are not affected by the removal of $T_\text{cut}$ or $T_\text{leaf}$.  
    It follows that $$(CC_i(T_\text{leaf},\s))_i > (CC_{\varphi(i)}(T_\text{cut},\s))_i,$$ where $\varphi$ is any permutation of $[n]$ that is the identity outside of component $K$ and swaps the agents in $T_\text{cut}$ and $T_\text{leaf}$. By \Cref{lemma:f-opp-increase}, we have $U_f(T_\text{cut}, \s) < U_f(T_\text{leaf}, \s)$.
\end{proof}

\noindent 
Also, selling edges in targeted regions makes the region less attractive for the \sqda{} opponent, i.e., it is edge-averse.

\corTechCor*
\begin{proof}
    Let $x$ be any node in $G(\s)$. Notice that $CC_x(T,\s)$ is the set of nodes $y$ such that there exists a path between $x$ and $y$ that does not contain nodes in $T$. In particular, this path contains no edge that was removed in profile~$\s'$. And because no new edges were created and no immunized agent sold immunity, it follows that $T' \subseteq T$. Therefore,
    the path between node~$x$ and node~$y$ also exists in $G(\s')$ when removing $T'$. Hence, $CC_x(T',\s) \ge CC_x(T,\s)$. which allows us to conclude by \Cref{lemma:f-opp-increase}.
\end{proof}

Before we prove our first structural property in \Cref{lemma:greed-drop}, we first show the following technical result for the function $g:(x,y) \mapsto f(x{+}y) - f(x) - f(y)$ (w.r.t function $f$ of the \sqda{} opponent) that is needed in \Cref{lemma:greed-drop}. This function computes the change to $U_f$ when the number of connected components change. Indeed, if a component of size $x{+}y$ gets split into components of size $x$ and~$y$, instead of adding a term $f(x{+}y)$ to $U_f$ it adds $f(x){+}f(y)$ instead. Hence this cut will change the value of $U_f$ by exactly $g(x,y)$.

\begin{restatable}{lemma}{lemmaGIncrease}\label{lemma:g-increase}
    The function $g:(x,y) \mapsto f(x+y) - f(x) - f(y)$ is strictly increasing in $x$, for all $x \geq 0$ and all fixed $y > 0$. 
\end{restatable}
\begin{proof}
    This is a direct consequence of the strict-convexity of $f$. Indeed by reordering terms, we have for all $y> 0$ and $x \ge 0$ the equality:
    \begin{align*}
    &\big[f(x{+}y{+}1){-}f(x{+}1){-}f(y) \big]\\
    &\quad\quad- \big[f(x{+}y){-}f(x){-}f(y) \big] \\
    =\;& \big[f(x{+}y{+}1){-}f(x{+}y) \big] - \big[f(x{+}1){-}f(x) \big]\\
    =\;& f'(x{+}y) - f'(x).
    \end{align*}
    We know that $f'$ is strictly increasing, by definition of strict convexity. Since we assume $y>0$, it holds that $x+y > x$ and thus that $f'(x{+}y) - f'(x)$ is strictly positive.
\end{proof}

The following lemma is our first structural result about equilibrium networks, where we state that under some assumptions, the loss of connectivity by dropping edges is less than the size of the component into which the edges were bought. This may seem obvious, especially since selling edges is supposed to make an agent less threatening to an \sqda{} opponent which should make it want to target something further away from the agent. However, as \Cref{example} shows, some non-trivial conditions are required for this to hold. Indeed, without edge~$e$ in profile~$\s'$, the connectivity of agent~$i$ is lowered by $9$, compared to having the edge in profile~$\s$, although the component $A$ only contains $7$ nodes.
\begin{restatable}{lemma}{lemmaGreedDrop}\label{lemma:greed-drop}
    Let $\s$ be a strategy profile with induced network $G(\s)$ that has a cut-vertex $u$ in some connected component~$Z$. We consider an $SQD$ opponet $\mathcal{A}$. Let $Z_v$ be a remaining connected component of $Z$ after removing node~$u$, and $v \in Z_v$ a node such that edge $\{u,v\}$ exists. Also, assume that there are either targeted regions outside of $Z$ or no targeted regions contained in $Z_v$. If $\s'$ is the profile where all edges between $u$ and a node in $Z_v$ are dropped, then $$\E_{\mathcal{T}(\s')}[CC_u(\s')] \ge \E_{\mathcal{T}(\s)}[CC_u(\s)] - |Z_v|.$$ Furthermore, if $u$ was targeted, then this inequality is strict.
\end{restatable}
\begin{figure}[h]
    \centering
    \includegraphics{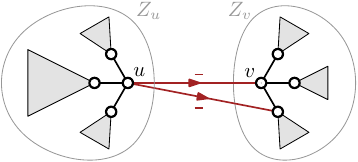}
    \caption{Illustration of the component $Z$ in \Cref{lemma:greed-drop}.}
    \label{fig:greed-drop}
\end{figure}

\begin{proof}
    We first rule out some simple cases, such that we can assume for the rest of the proof, that there is at least one targeted region in $G(\s)$ that does not include node~$u$.

    \begin{claim}\label{claim:greed-drop1}
        The statement holds if all targeted regions that exist in $G(\s)$ contain node~$u$.
    \end{claim}
    \begin{claimproof}[\Cref{claim:greed-drop1}]
        This condition is only true if there are either no targeted regions in $G(\s)$ at all or the only targeted region is the vulnerable region containing node~$u$.

        If $\T(\s)$ is empty, we know that all nodes in $G(\s)$ are immunized. In this case, agent~$u$ looses exactly $|Z_v|$ of connectivity by selling all edges to $Z_v$, i.e., $\E_{\mathcal{T}(\s')}[CC_u(\s')] \ge \E_{\mathcal{T}(\s)}[CC_u(\s)] - |Z_v|$. In this case, node~$u$ is immunized and cannot be targeted.

        If, on the other hand, $\T(\s)$ is not empty, we know by assumption that the only targeted region in $G(\s)$ contains node~$u$. Then, it holds that $\E_{\mathcal{T}(\s)}[CC_u(\s)] = 0$ and $\E_{\mathcal{T}(\s')}[CC_u(\s')] \geq 0$. Thus, 
            $$\E_{\mathcal{T}(\s')}[CC_u(\s')] > \E_{\mathcal{T}(\s)}[CC_u(\s)] - |Z_v|$$ 
        clearly holds.
    \end{claimproof}

    By \Cref{claim:greed-drop1}, we can assume for the rest of this proof that there is at least one targeted region in $\s$ that does not contain node~$u$. Next, we prove that node~$u$ is not targeted in profile~$\s'$, which will help us in the remainder of the proof.

    \begin{claim}\label{claim:greed-drop2}
        The node~$u$ is not targeted in strategy profile $\s'$.
    \end{claim}
    \begin{claimproof}[\Cref{claim:greed-drop2}]
        From the previous considerations, we know that there is at least one targeted region $T \in \mathcal{T}(\s)$ that does not contain node~$u$. The region $T$ is also a vulnerable region in strategy profile~$s'$, since the only difference between $\s$ and $\s'$ are edges of agent~$u$ which have no connection to $T$ (by definition of $T$). To show that node~$u$ is not targeted in $\s'$, we show that
            $$U_f(T,\s') \overset{(1)}{<} U_f(T,\s) \overset{(2)}{\leq}  U_f(u,\s) \overset{(3)}{\leq} U_f(u,\s').$$
        Since the \sqda{} opponent $\mathcal{A}$ only attacks nodes / vulnerable regions that minimize $U_f$, node~$u$ is then not targeted in $\s'$.

        For Inequality (1), we use \Cref{lemma:f-opp-increase}. Because removing the edges only lowers the size of the connected components, and some of these became strictly lower as nodes~$u$ and $v$ got disconnected. This gives the desired result.
        Inequality (2) holds due to the fact that $T$ is a targeted region in $\s$ and thus minimizes $U_f$.
        Inequality (3) is a direct consequence of \Cref{cor:tech-cor2}, as the edges sold were all incident to the agent $u$ that was in the considered vulnerable region. 
    \end{claimproof}

    In the following, let $Z_u \coloneqq Z \setminus Z_v$ be the connected component $Z$ of $G(\s)$ without $Z_v$. To show $\E_{\mathcal{T}(\s')}[CC_u(\s')] \ge \E_{\mathcal{T}(\s)}[CC_u(\s)] - |Z_v|$, we will consider two different cases: In profile $\s$, either there is a targeted region outside of $Z$ (Case 1) or all targeted regions are inside $Z$ (Case 2).\\[-2pt]

    \begin{claimproof}[Case 1]
        If there is at least one targeted region~$T$ (w.r.t. $\s$) outside of $Z$, then we show that all targeted regions (w.r.t. $\s'$) are outside of $Z_u$. Once this is shown, we know that every attack will not affect the connected component of node~$u$ in $G(\s')$. Therefore, it holds that $$\E_{\mathcal{T}(\s')}[CC_u(\s')] = |Z_u| = |Z| - |Z_v|.$$ Additionally, we know that in strategy profile~$\s$, the connectivity of agent~$u$ is at most the size of its connected component, i.e., $\E_{\mathcal{T}(\s)}[CC_u(\s)] \le |Z|$. With this, we get $\E_{\mathcal{T}(\s')}[CC_u(\s')] \ge \E_{\mathcal{T}(\s)}[CC_u(\s)] - |Z_v|$.

        It remains to show that indeed all targeted regions in $G(\s')$ are outside of $Z_u$, if there is a targeted region $T$ in $G(\s)$ outside of $Z_u$. For this, let $T' \subseteq Z_u$ be a targeted region in $Z_u$ with respect to profile~$\s'$. By Claim 1, we know that $u$ is not targeted in $G(\s')$ and thus $u \notin T'$. We will show that $U_f(T,\s') < U_f(T', \s')$. Since removing $T$ does not affect $Z$, we know that
        \begin{align}
            U_f(T,\s) = U_f(T,\s') + g(|Z_u|,|Z_v|),\label{eq:greed-drop:case1}
        \end{align}
        where $g$ is the function from \Cref{lemma:g-increase}. To start the chain of inequalities, we rephrase this equation to be
            $$U_f(T,\s') = U_f(T,\s) - g(|Z_u|,|Z_v|)$$
        Because $T$ is a targeted region in strategy profile~$\s$, we get that $U_f(T,\s) \leq U_f(T',\s)$. Additionally, from \Cref{lemma:g-increase}, we know that $g(x,y)$ is strictly increasing for fixed $y > 0$. As at least $v \in Z_v$, we can apply this statement to get that 
            $$g(|Z_u|,|Z_v|) > g(|Z_u| - |T'|,|Z_v|).$$
        As a result, it holds that 
        \begin{align*}
            U_f(T,\s') =\;& U_f(T,\s) - g(|Z_u|,|Z_v|).\\
            <\;& U_f(T',\s) - g(|Z_u| - |T'|,|Z_v|).\\
            =\;& U_f(T',\s'),
        \end{align*}
        where the last equality holds due to similar considerations as in \Cref{eq:greed-drop:case1}.
        
        Since the \sqda{} opponent only targets nodes that minimize $U_f$, $T'$ (and thus all vulnerable regions in $Z_u$ w.r.t $s'$) are not targeted.
    \end{claimproof}
    
    \begin{claimproof}[Case 2]
        Consider that all targeted regions are in connected component~$Z$. 
Consider $$W = \mathrm{argmax}_{T \in \T(\s)} CC_u(T, \s).$$ Notice that $W$ is in $Z$, as assumed in Case 2 for every targeted region. Additionally, $W$ is not included in $Z_v$, as assumed in the statement of the lemma, and it does not contain node-$u$. Since we are not in the setting handled by Claim 1, at least one targeted region does not contain node-$u$ and because $W$ maximizes agent $u$'s connectivity, also $W$ does not contain node~$u$. Therefore, we have $W \subsetneq Z_u$.
Hence, it holds that $CC_u(W, \s) \ge E_{\T}[CC_u(\s)]$. Consider $T' \in \T(\s')$. Because $u$ is not targeted, one of the following conditions must be met:
\begin{enumerate}[leftmargin=1cm]
    \item[(a)] $T' \cap Z_u = \emptyset$,
    \item[(b)] $T' \subsetneq Z_u$.
\end{enumerate}

\noindent
In case (a), we can state that $$CC_{u}(T',\s') = |Z| - |Z_v| \geq \E_{\T}[CC_u(\s)] - |Z_v|.$$ And if node~$u$ was targeted then $\E_{\T}[CC_u(\s)] < |Z|$ thus $CC_{u}(T',\s') > \E_{\T}[CC_u(\s)] - |Z_v|$.

For case (b), we show the following chain of inequalities:
\begin{align*}
    CC_u(T',\s') \overset{(i)}{\geq}\;& CC_u(W,\s')\\
    \overset{(ii)}{\geq}\;& CC_u(W,\s) - |Z_v|\\
    \overset{(iii)}{\geq}\;& \E_{\T(\s)}[CC_u(\s)] - |Z_v|,
\end{align*}
for every $T' \in \T(\s')$. As this establishes a lower bound for $CC_u(T',\s')$ over all targeted regions in $G(\s')$, also the expected values of $CC_u(T',\s')$, i.e., the connectivity of $u$ w.r.t profile~$\s'$, is bounded below by $\E_{\T(\s)}[CC_u(\s)] - |Z_v|$. Additionally, if node~$u$ is targeted in profile~$\s$, Inequality (iii) becomes strict, which concludes the proof. 

Now we show the above inequalities. Inequality $(ii)$ follows, since agent~$u$ removes all edges to $Z_v$ in profile $\s'$.

For Inequality ($i$), since $W \in \T$, we know that the inequality $U_f(W, \s) \leq U_f(T',\s)$ holds and similarly, because $T' \in \T(\s')$ also $U_f(W, \s') \geq U_f(T',\s')$ holds. This gives us the result 
\begin{equation}U_f(W,\s)-U_f(W,\s') \le U_f(T',\s) - U_f(T',\s').\label{eq:Uf}\end{equation}
Now we can compute how these values change when changing the strategy profile from $\s$ to $\s'$ and keeping the targeted region fixed. Note that because $W \subsetneq Z_u$ and $T' \subsetneq Z_u$, going from strategy profile $\s$ to $\s'$ disconnects two components, which change the value $U_f$ by some output of the function~$g$. This gives us: $$U_f(W, \s) - U_f(W, \s') = g(CC_v(W, \s'),CC_u(W, \s'))$$ and that $$U_f(T', \s) - U_f(T', \s') = g(CC_v(T', \s'), CC_u(T', \s')).$$ But it can be simplified, given that $CC_v(W, \s') = CC_v(T', \s') = |Z_v|$, since $W \subseteq Z_u$. Inequality~(\ref{eq:Uf}) and using, once again, that function~$g$ is strictly increasing (by \Cref{lemma:g-increase}) and, by the contrapositive of this, we conclude that $CC_u(W, \s') \leq CC_u(T', \s')$. 

Inequality $(iii)$ holds, since $W$  a maximizer of $CC_u$, we know that $\E_{\T}[CC_u(\s)] \leq CC_u(W,\s)$.
\end{claimproof}

This finishes the proof.
\end{proof}
\noindent Note that the condition of \Cref{lemma:greed-drop} about the targeted regions is necessary, as \Cref{example} shows, with node $i$ acting as node $u$ and node $j$ acting as node $v$.  

The next corollary is one of the main tools for our other results. It follows from \Cref{lemma:greed-drop} (see Supp. Material).

\corSizeAfterEdge*

\begin{proof}
    This is a direct application of \Cref{lemma:greed-drop} when looking at the best response of agent~$u$. The fact that it does not strictly prefer to drop the $k$ edges to $Z$ while it loses at most $|Z|$ (respectively, strictly less if $u$ is targeted) in connectivity means that the cost is lower (respectively, strictly lower) than the size of $Z$.
\end{proof}

\lemmaTargetsSingleton*
\begin{proof}
    Assume towards a contradiction that there exists a connected component $Z$ with at least one immunized node and a targeted region $T$ such that $|T| > 1$. Let agent~$x$ be a node of $T$ that shares an edge with an immunized node $y$. \Cref{lemma:1} ensures that $T$ is a tree. We can root $T$ in node~$x$ and take any leaf $u$ of this tree. Because $|T| > 1$, there is an edge between node~$u$ and some other node~$v$ in $T$. We claim the agent among $u$ and $v$ that bought the edge strictly prefers to swap it to some immunized node in $Z$. 

    First, region~$T$ cannot be the only targeted region as agents~$u$ or $v$ would have a negative utility, and thus would prefer to drop all edges. Therefore, there is at least one other targeted region~$T'$. Also, because $T$ is a tree, it would split into two subtrees~$T_u$ and $T_v$ if the edge was dropped. This allows us to use \Cref{cor:size-after-edge} on node~$v$ to conclude that if it bought the edge, then node~$u$ shares edges with immunized nodes. (If it does not, then no targeted regions are included in $\{u\}$ and $C_E>1$.) And if agent~$u$ bought the edge, there is, by construction, an edge between immunized nodes and $T_v$ (the edge $\{x,y\}$).
    
    Therefore, in both cases, there is an agent that bought an edge, such that the vulnerable component on the other side of the edge is linked to an immunized node. Without loss of generality, we assume that agent~$u$ is the buyer (but the following deviation would also hold for agent~$v$).
    
    Now consider the strategy change of agent~$u$ swapping its edge to $v$ for an edge to $y$ instead. See \Cref{fig:singleton_targeted}. 
    \begin{figure}[h]
    \centering
    \includegraphics[width=0.38\linewidth]{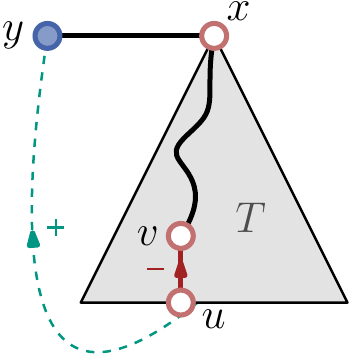}
    \caption{Illustration of \Cref{lemma:targetsingletons}. The tree $T$ rooted in $x$ with leaf $u$. Red edges (supplemented with \enquote{-}) and green edges (supplemented with \enquote{+}) denote which edges are sold and bought, respectively, in profile $\s'$ compared to profile $\s$.}
    \label{fig:singleton_targeted}
\end{figure}
    In that case, the value of $U_f$ when targeting $T_u$ would be strictly higher, as $T_v$ is not destroyed anymore, and the same argument applies for targeting $T_v$. However, when targeting any other vulnerable region, the value of $U_f$ does not change. Therefore, the new targeted regions are the previous ones except for $T$.
    Hence, the new utility of agent~$u$ is the previous expected utility conditioned on the fact that region~$T$ is not destroyed, which is a strict improvement.
\end{proof}

\lemmaConditionCutTargeted*

\begin{proof}[Proof of i)]
    Let $\s$ be a Nash equilibrium and $Z$ be a connected component $G(\s)$ with at least one immunized node and a targeted cut-vertex. First, we notice that every targeted node in $Z$ must be a cut-vertex, because of \Cref{cor:sqd-cut-vertices} and the fact that all targeted regions are singletons (see \Cref{lemma:targetsingletons}). 
    
    Let $y$ be a targeted cut-vertex such that removing $y$ creates a new connected component $Z_x$ of minimum size and some other component $Z_w$. Let $x$ be a node in component~$Z_x$ such that node~$x$ and node~$y$ share an edge, and let $w \in Z_w$ be an immunized node that shares an edge with $y$. This always exists because node~$y$ is assumed to be a cut-vertex and because targeted regions are singletons. See \Cref{fig:targeted_cut}. 
    \begin{figure}[h]
    \centering
    \includegraphics[width=\linewidth]{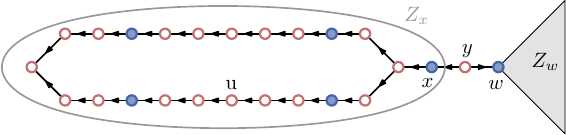}
    \caption{Illustration of the proof of \Cref{lemma:condition-cut-targeted}. The studied component $Z$. The considered strategy changes are either $x$ or $u$ swapping an edge for one to node~$w$ instead.}
    \label{fig:targeted_cut}
\end{figure}

    We will first show that the size of $Z_x$ is strictly greater than $C_E$, by using \Cref{cor:size-after-edge}. However, we first have to prove that we respect the assumptions of the corollary.
    \begin{claim}\label{claim:condition-cut-targeted:1}
        Every edge between a node in $Z_x$ and node~$y$ is bought by agent~$y$.
    \end{claim}
    \begin{claimproof}[\Cref{claim:condition-cut-targeted:1}]
    Because we have not made any assumptions on node~$x$ except that it is a node in $Z_x$ that shares an edge with agent~$y$, it is enough to prove that this edge was bought by agent~$y$.

    By contradiction, assume that agent~$x$ bought the edge. Look at the strategy change of swapping it for edge~$\{x,w\}$ instead. Because agent~$x$ is immunized, it cannot become targeted, and the sizes of the connected components do not change when targeting any region that does not contain node~$y$, and it increases strictly when targeting node~$y$ (which allows us to use \Cref{lemma:f-opp-increase}). 
    This allows us to investigate how targeted regions change by this deviation; two cases for this change will have an easy conclusion and the last one will be split into three sub-cases. 
    
    If there were targeted regions other than $\{y\}$, they stay targeted and node~$y$ is not targeted anymore, which strictly increases the utility of agent~$x$.
    
    For the second case, if $\{y\}$ was the only targeted region and it stays the only targeted region: this is also strictly improving for agent~$x$.
    
    Finally, if $\{y\}$ was the only targeted region but new targeted regions are created by the deviation, there are three possibilities for their position:
    \begin{enumerate}
        \item If a new targeted region is outside of $Z$, then the utility of agent~$x$ strictly increases when this region is attacked.
        \item If a new targeted region is in $Z$ but not in $Z_x$, then agent~$x$'s utility also increases because its corresponding connected component contains at least the previous one and the nodes~$\{y,w\}$ when the new region is attacked.
        \item If a new targeted region is in $Z_x$, then because of the minimality assumption, agent~$x$'s utility is strictly higher when this region is attacked, because its new connected component contains $Z_w$ and node~$y$, and by minimality assumption, $|Z_x| \le |Z_w|$.
    \end{enumerate}
    So this strategy change is always strictly increasing for agent~$x$, which contradicts that strategy profile~$\s$ is a Nash equilibrium. Therefore, we can conclude that the edge $\{y,x\}$ was bought by agent~$y$.
    \end{claimproof}
    
    Because $y$ bought an edge, it must not be the only targeted region (or it would prefer to drop the edge) and by minimality assumption, no targeted regions are in $Z_x$. We are, therefore, in the situation of \Cref{cor:size-after-edge}, where the cut-vertex is targeted. That means that $Z_x$ contains strictly more than $C_E$ vertices. And because of the minimality assumption, we have proven the first assertion.
\end{proof}

\begin{proof}[Proof for ii)]
    We use the same notation and the same strategy profile $\s$ as in the proof for $i)$, and also the derived results. With this, we now prove that at least two of the targeted regions are outside of component~$Z$. To do so, notice that $|Z_x| > C_E > 1$ and therefore there exists an edge between two vertices in $Z_x$. Let node~$u$ be the agent that bought this edge, and consider the same swap as before, dropping it for the edge~$\{u,w\}$ instead. This yields the new strategy profile~$\s'$, which will allow us to prove that there is at least one targeted region outside of component~$Z$. Then, by considering another deviation, we will show that two or more such targeted regions outside of $Z$ exist. 
    
    Like in the proof of i), agent~$u$ swapping its edge to edge~$\{u,w\}$ dose not change the value of $U_f$ when targeting any region outside of $Z_x$ and strictly increases $U_f$ when targeting $y$. Therefore, node~$y$ cannot be targeted after the deviation, as we proved in part i) that it was not the only targeted region. However, new targeted regions can appear in $Z_x$ when doing so. 
    
    Let $T' \in \T(\s')$ be one of those newly created vulnerable regions. If node~$u$ is in $T'$, then the value of $U_f$ is greater or equal than when targeting node~$u$ before the deviation, so $T'$ cannot be targeted. Therefore, we know that no newly created targeted region includes node~$u$. To show that there is at least one targeted region outside of component~$Z$, we show that not all targeted regions are inside $Z$. Indeed, assuming they were all in $Z$, we can prove that the deviation of agent~$u$ swapping its edge to node~$w$ is still strictly improving for agent~$u$. 
    
    Let $T \in \T(\s)$ be any previously targeted region. By assumption, region~$T$ was a cut-vertex in component~$Z$ and targeting it splits $Z$ in at least two connected components, both of them being larger than $Z_x$, by minimality. Hence, both of them are smaller than $|Z| - 1 - |Z_x|$. Therefore, it holds that $$\E_{\T}[CC_u(\s)] \le |Z| - 1 - |Z_x|.$$ But after the edge-swap, when targeting any region~$T' \subseteq Z_x$ that does not contain node~$u$, we have $$CC_{u}(T', \s') > |Z| - |Z_x| > \E_{\T}[CC_u(\s)].$$ Thus, the deviation strictly increases the utility of agent~$u$. Indeed, the new regions that were possibly added all lead to more connectivity for agent~$u$, so they improve its utility. This means that because node~$y$ is not targeted anymore, the utility of agent~$u$ is strictly higher in $\s'$ than in $\s$. 

    Therefore, there is at least one targeted region $T_o \in \T(\s)$ that lies outside of $Z$. For the sake of contradiction, assume it was the only one. Then, the utility of agent~$y$ in $\s$ is
    \begin{align*}
    \E_{\T(\s)}[CC_y(\s)] &= \frac{1}{|\T(\s)|} \bigg[|Z| + \sum_{T \in \T(\s) \backslash \{T_o, y\}} CC_y(T,\s) \bigg]  \\ 
    &\le \frac{1}{|\T(\s)|} \bigg[|Z| + (|\T(\s)| - 2) (|Z| - |Z_x|) \bigg] \\
    &\le \frac{|\T(\s)| - 1}{|\T(\s)|}(|Z| - |Z_x|) + \frac{1}{\T(\s)}|Z_x| \\
    &<|Z| - |Z_x|.
    \end{align*}
    The last inequality comes once again from the minimality of $|Z_x|$, as $|Z|-|Z_x| \ge |Z_w| + 1 \ge |Z_x|+1$. However, after agent~$y$ removes the edge $\{y,x\}$, the only targeted region becomes $T_o$ (indeed, $U_f$ decreases strictly less when targeting regions in $Z$ after the removal of $\{y,x\}$ compared to before, while the decrease when targeting $T_o$ does not change).
    
    Therefore, the new utility is $$\E_{\T(\s')}[CC_y(\s')] = |Z| - |Z_x|.$$
    Hence, this deviation would be strictly improving, which is a contradiction. Thus, we have shown that there are at least two targeted regions outside of component $Z$.
\end{proof}

\lemmaNoIsolatedCut*
\begin{proof}
    Assume towards a contradiction, that network~$G(\s)$ contains an isolated node~$x$ as well as a connected component~$Z$ with an immunized node and a targeted cut-vertex $y$. Let $Z_1, Z_2$ be two remaining connected components after removing node~$y$. 
    
    Now, consider the deviation of agent~$x$ buying an edge to an immunized node in $Z_1$ that shares an edge with $y$. Let $\s'$ be the resulting strategy profile. We will show that agent~$u$ prefers profile~$\s'$ over profile $\s$, which is a contradiction to $\s$ being an equilibrium.

    Note that agent~$x$ is not targeted in profile~$\s'$ because, by \Cref{cor:sqd-cut-vertices}, the \sqda{} opponent $\mathcal{A}$ still prefers to attack node~$y$.

    Since agent~$x$ is not targeted in profile~$\s'$, even after the attack its connected component contains at least either $Z_1$ or $Z_2$. We know that $|Z_1| > C_E$ and $|Z_2| > C_E$ hold by \Cref{lemma:condition-cut-targeted}, because profile $\s$ is an equilibrium. In profile~$\s'$, agent~$x$'s connected component post-attack thus has size strictly greater than $C_E+1$, which yields a strictly better utility for agent~$x$ than staying isolated.
\end{proof}
\noindent
Next, we bound the sizes of connected components.
\lemmaMinimalSize*
\begin{proof}
    Let $Z$ be a connected component in $G(\s)$ with a size larger than $1$. We will show that the size of $Z$ is at least $C_E + 1$. For this, let $u \in Z$ be an agent that bought an edge, and consider the strategy change of agent~$u$ dropping every bought edge and immunization if it was bought. Let $\s'$ be the corresponding strategy profile. We will compare the utility of agent~$u$ in both profiles $\s$ and $\s'$ and use the fact that $\s$ is a Nash equilibrium to deduce an inequality that helps us to prove $|Z| \geq C_E+1$. For this, we make a case distinction on whether agent~$u$ is vulnerable in strategy profile $\s$.

    If $u$ is vulnerable in profile~$\s$, then, by using \Cref{cor:tech-cor2}, the probability of being destroyed does not increase against an \sqda{} opponent $\mathcal{A}$ after the deviation. Let $p$ and $p^\prime$ denote the probability of agent~$u$ being destroyed in $\s$ and $\s'$, respectively. We have shown that $p^\prime \leq p$. The utility of agent~$u$ in profile~$\s$ is lower than $(1-p)|Z| - C_E$, while agent~$u$'s utility in profile~$\s'$ is higher than $1-p' \ge 1-p$. As profile~$\s$ is a Nash equilibrium, agent~$u$ weakly prefers its strategy in $\s$ over its strategy in $\s'$, i.e., it holds that $(1-p) (|Z|-1) \ge C_E$. Finally, solving for $|Z|$ and using $p \leq 1$ yields $|Z| \geq C_E + 1$.
        
    If $u$ is immunized in profile~$\s$, its utility in $\s$ is at most $$|Z|-C_E-C_I < |Z|-C_E-1$$ and its utility in $\s'$ is at least $0$. Therefore the equilibrium property gives us $|Z| \ge C_E+1$. 
\end{proof}

\noindent 
The next lemma deals with the implications of the existence of a targeted region that is a cut-vertex inside a component with an immunized nodes.

\singleComponentCut*
\begin{proof}
    Let $Z_1$ and $Z_2$ be the two connected components, such that both contain an immunized node and a targeted region. Assume for contradiction, that (w.l.o.g) component~$Z_1$ contains a targeted cut-vertex. If $Z_2$ also contains a targeted cut-vertex, without loss of generality we can assume $|Z_1| \le |Z_2|$. Notice that even if only $Z_1$ has targeted cut-vertices, this inequality also holds. Indeed, if $v$ is a targeted cut-vertex in $Z_1$ and $w$ a targeted node in $Z_2$, it holds that $U_f(v,\s) = U_f(w,\s)$ or one of them would not be targeted. However, it holds that $$U_f(\emptyset, \s)-U_f(w,\s) = f'(|Z_2|-1),$$
    and that
    $$U_f(\emptyset,\s) - U_f(v,\s) = f(|Z_1|) - \sum_{i}f(|W_i|),$$
    where $(W_i)_i$ is the vector of the remaining connected components of $Z_1$ after removing node~$v$. Because $v$ is a cut-vertex, there are at least two such components which allows us to use the strict convexity of $f$ to get
    \begin{align*}
        U_f(\emptyset,\s) - U_f(v,\s) >\;& f(|Z_1|) - f\left(\sum_i |W_i|\right)\\
        =\;& f'(|Z_1|-1).
    \end{align*}
    Hence, we have $f'(|Z_1|-1) < f'(|Z_2|-1)$ which gives us $|Z_1| < |Z_2|$ by using the convexity of $f$.
    
    We will show that every immunized node $x \in Z_1$ has a strictly improving strategy change. For this, consider the deviation where agent~$x$ buys an edge to an immunized node $y \in Z_2$ that shares an edge with a targeted node $u \in Z_2$. Let $\s'$ be the corresponding strategy profile.
    
    In profile $\s$, the connectivity $\E_{\T(\s)}[CC_x(\s)]$ of agent~$x$ is strictly lower than $|Z_1|$ because we assumed that there is a targeted region inside $Z_1$. To show that agent~$x$ prefers its strategy in $\s'$ over its strategy in $\s$, we have to prove that agent~$x$'s gain in connectivity is strictly larger than the edge cost $C_E$, i.e.,
        $$\E_{\T(\s')}[CC_x(\s')] - \E_{\T(\s)}[CC_x(\s)] > C_E.$$
    For this, we show that $$CC_x(T,\s') \geq |Z_1| + C_E$$ for all targeted regions. 
    For targeted regions outside of $Z_1$ and $Z_2$ this holds, because $|Z_2| \ge C_E+1$.
    
    If a targeted region $T$ is in $Z_1$, then, when removing it, the component of node~$x$ contains $Z_2$ and $CC_x(T,\s)$ whose sizes are greater than $|Z_1|$ and $C_E$, respectively. This holds, since by \Cref{lemma:condition-cut-targeted} we have that $CC_x(T,\s)>C_E$ and we have already shown that $|Z_1| \leq |Z_2|$.

    If region~$T$ is in $Z_2$ instead, the component of node~$x$ contains $Z_1$ and $CC_y(T,\s)$. The size of this union is greater or equal than $|Z_1| + C_E$. Indeed, if node~$u$ is not a cut-vertex, then no vulnerable cut vertices exist in $Z_2$ and it holds from \Cref{lemma:minimal_size}, else it is implied by \Cref{lemma:condition-cut-targeted}. 
\end{proof}
\noindent 
Later will be bound the number of agents in some cases of Nash equilibria. The following statement is useful for this.

\lemmaBoundedNumberComponentsCut*
\begin{proof}
    Let $Z_1,\dots,Z_r$ be the connected components in equlibrium network~$G(\s)$ with $|Z_1| \geq \dots \geq |Z_r|$. If $r\leq 2$, then we are done. Thus, assume that $r>2$.
    Now, consider any node~$x$ in the connected component~$Z_r$ of smallest size and consider the deviation of agent~$x$ immunizing (if it were not already immunized) and buying an edge to any node in every other component. Let $\s'$ be the corresponding strategy profile. We will compare the utilities of agent~$x$ in both profiles $\s$ and $\s'$ and use the fact that $\s$ is a Nash equilibrium to deduce an inequality that implies $C_E + C_I +2 \geq r$. 
    
    Let $k$ be the number of edges bought by agent~$x$ in profile~$\s$. The utility of agent~$x$ in profile~$\s$ is at most $|Z_r| - k C_E$ and agent~$x$'s utility in profile~$\s'$ is at least 
        $$\sum_{i=1}^r |Z_i| - \max_{1\le i \le r} |Z_i| - (k+r-1) C_E - C_I.$$
    Since profile~$\s$ is a Nash equilibrium, we know that agent~$x$ weakly prefers its strategy in $\s$ over its strategy in $\s'$, and thus 
    \begin{align*}
        &|Z_r|-kC_E
            \geq\;  \sum_{i=2}^r |Z_i| - (k+r-1) C_E - C_I.
    \end{align*}
    Now, due to \Cref{lemma:minimal_size} and the fact that there are no isolated nodes, $|Z_i| \ge C_E +1$ for all $i$. Therefore, it holds that
        $$|Z_r| - kC_E \geq\; |Z_r| + (r-2)(C_E+1) - (k+r-1)C_E - C_I.$$
    Reordering the terms give us what we wanted: \begin{align*}
        C_I + C_E +2 &\ge r. \qedhere
    \end{align*} 
\end{proof}

\noindent
The next lemma allows us to eliminate the possibility of targeted cut-vertices as they only exist in small networks and thus do not impact the asymptotic bounds.

\lemmaCutImplyBounded*
\begin{proof}
    The previous lemmas gave us strict conditions for the existence of such an equilibrium. By definition there is a connected component~$Z$ in $G(\s)$ with an immunized node and a targeted cut-vertex. \Cref{lemma:condition-cut-targeted} ensures the existence of some other targeted regions $T_1,...,T_k$ with $k \ge 2$ outside of $Z$. None of these targeted regions lie inside a component with an immunized node or it would contradict \Cref{single-component-cut}. Therefore, all targeted regions are either contained in $Z$ or form their own connected component. Based on that, we can make some observations regarding the components $T_1,\dots,T_k$:

    \begin{enumerate}[label=(\roman*),leftmargin=32pt]
        \item All components $T_1,\dots,T_k$ have the same size, since all are targeted regions at the same time.
        \item We know that $|Z|\ge|T_i|$ for all $i \in [k]$. Otherwise, if $|Z| \leq |T_i|$ for some $i \in [k]$, no node in $Z$ would be targeted. This is because $f$ is increasing because it is convex, has positive values and starts at $0$.
    \end{enumerate}

    \noindent
    In \Cref{lemma:bounded-number-components-cut}, we found an upper bound on the number of connected components in any Nash equilibrium without isolated nodes, which applies here as \Cref{lemma:no-isolated-cut} shows that there are none in this case. The goal of this proof is to bound the number of nodes in these components. For this, we proceed in three steps, each of which determines an upper bound on the size of a specific structure in the equilibrium. The structures are:\footnote{Technically, the second two cases cover all kinds of connected components and thus all nodes. However, we need the first step as an intermediate result for the steps 2 and 3.}
    \begin{enumerate}[leftmargin=32pt]
        \item the targeted regions $T_i$,
        \item the other components without immunized centroid,
        \item the other components with an immunized centroid.
    \end{enumerate}

    \begin{step}\label{step1:cut-imply-bounded}
        The size of every targeted region is bounded above by $2C_E+C_I$.
    \end{step}
    \begin{claimproof}[\Cref{step1:cut-imply-bounded}]
        Take any agent $x \in T_1$ and consider the deviation of agent~$x$ immunizing and buying an edge to any node in $T_2$ and to any immunized node of $Z$. Before the deviation, the utility of agent~$x$ was at most $|T_1|-\ell C_E$, where $\ell \geq 0$ is the number of edges bought by $x$ into its component $T_1$. After the deviation, node~$x$ is connected to at least two whole components among $Z$, $T_1$ and $T_2$ no matter which region was attacked. Therefore, agent $x$'s utility after the deviation is at least
        \begin{align*}
            &2\min\{|Z|,|T_1|,|T_2|\} - (\ell +2)C_E - C_I \\\geq\;& 2|T_1| - (\ell +2)C_E - C_I.
        \end{align*}
        Using the Nash equilibrium property that this deviation is no strict improvement for agent~$x$, we get that $|T_1| \leq 2C_E+C_I$.
    \end{claimproof}

    \begin{step}\label{step2:cut-imply-bounded}
        Every component without an immunized centroid has at most $\sqrt{2}(2 C_E + C_I)$ vertices.
    \end{step}
    \begin{claimproof}[\Cref{step2:cut-imply-bounded}]
        For this proof, let $K$ be a connected component without immunized centroid. \Cref{lemma:no-centroid-implies-vulnerable} ensures that in that case, there is a vulnerable region $T \subseteq K$ that if removed, splits $K$ in connected components of size $p_1,...,p_r$ such that $0 <p_i \le |K|/2$, for all $1\le i \le r$. This will allow us to prove the upper bound by claiming that the opponent would strictly prefer targeting $T$ over targeting $T_1$, if the size of $K$ was greater than $\sqrt{2}(2 C_E + C_I)$.

        First, because the condition implies that $|K| - 2 p_i \ge 0$ for all $i$, we have that $$\sum_{i=0}^r p_i(|K| - 2 p_i) \ge 0.$$ Then by expanding the terms, we get:
        \[
         |K| ^2 - 2 \sum_{i=1}^r p_i^2 \ge |K| \sum_{i=1}^r p_i - 2 \sum_{i=1}^r p_i^2 \ge 0.
        \]
        Which means that 
        \begin{equation}\sum_{i=1}^r p_i^2 \le \frac{|K|^2}{2} \label{eq_squares}
        \end{equation} 
        holds. By rearranging Inequality~(\ref{eq_squares}) we get 
        $$|K^2| - \sum_{i=1}^r p_i^2 \geq \frac{|K|^2}{2},$$ which bounds the damage done by the maximum disruption opponent when attacking region~$T$. We now transfer this lower bound on the damage done to the more general class of \sqda{} opponents. For this  we use that \sqda{} opponents are $f$-opponents where $f/x^2$ is non-decreasing. For this, we can multiply and divide by the square function to turn this $f(|K|) - \sum_{i=1}^r f(p_i)$ into:
        \begin{align*}
            &\frac{f(|K|)}{|K|^2} |K|^2 - \sum_{i=1}^r \frac{f(p_i)}{p_i^2}p_i^2\\
            \ge\;& \frac{f(|K|)}{|K|^2}\left(|K|^2 - \sum_{i=1}^rp_i^2\right)\\
            \ge\;& \frac{f(|K|)}{2}.
        \end{align*}
        Hence, the fact that $T_1$ is targeted ensures that $f(|T_1|) \ge \frac{f(|K|)}{2}$. Now we will use the trick of multiplying and dividing by $x^2$ once again. If we assume towards a contradiction, that $|K| > \sqrt{2} |T_1| > |T_1|$, then we can use the non-decreasing property :  $\frac{f(|K|)}{|K|^2} \ge \frac{f(|T_1|^2)}{|T_1|^2}$. Therefore, if $|K| > \sqrt{2} |T_1|$ then:
        \begin{align*}
        U_f(T,\s) - U_f(\emptyset, \s) &\ge \frac{f(|K|)}{2} \\
         &= \frac{f(|K|)}{|K|^2}\frac{|K|^2}{2} \\ &\ge\frac{f(|T_1|)}{|T_1|^2}\frac{|K|^2}{2} \\
        &> \frac{f(|T_1|)}{|T_1|^2}|T_1|^2 \\
        &= f(|T_1|) \\
        &= U_f(T_1,\s) - U_f(\emptyset, \s).
         \end{align*}
        Which contradicts the fact that $T_1$ is targeted. Thus, $|K| \le \sqrt{2} |T_1| \le \sqrt{2}(2 C_E + C_I)$.
    \end{claimproof}

    \begin{step}\label{step3:cut-imply-bounded}
        Every component with an immunized centroid has at most $2 C_I + 3 C_E$ nodes.
    \end{step}
    \begin{claimproof}[\Cref{step3:cut-imply-bounded}]
        Take any component~$K$ with an immunized centroid. Look at the deviation of any agent~$x \in T_1$ to immunize and buy an edge to this centroid. This increases the cost of agent~$x$ by $C_E + C_I$. However, if a region in $T_1$ is removed, the size of the connected component of node~$x$ is at least $|K|$, and if a region in $K$ is removed, this size is at least $|T_1| + \frac{|K|}{2}$ because of the centroid property. Hence, for this deviation to not be strictly increasing, either $|K| \le |T_1| + C_I + C_E$ or $|K| \le 2(C_I +C_E)$ must hold. Since, $|T_1| \leq 2C_E+C_I$, by \Cref{step1:cut-imply-bounded}, both of them imply that $|K| \le 2C_I + 3 C_E$, we have the desired bound.
    \end{claimproof}

    \vspace{6pt}
    \noindent
    To summarize, there are at most $C_E+C_I+2$ connected components (by \Cref{lemma:no-isolated-cut} and \Cref{lemma:bounded-number-components-cut}), each of which have a size of at most $2C_I + 3 C_E$ nodes (if they contain an immunized centroid) or $\sqrt{2}(2 C_E + C_I)$ nodes (if they don't). Since $2C_I + 3 C_E > \sqrt{2}(2 C_E + C_I)$ for all $C_E,C_I > 1$, we get that the given Nash equilibrium network has at most $(C_E+C_I+2)(2C_I + 3 C_E)$ nodes.
\end{proof}

\lemmaTargetedWithImmunized*
\begin{proof}
    Assume towards a contradiction that there is a non-trivial Nash equilibrium $\s$ with respect to an \sqda{} opponent $\mathcal{A}$, where no targeted region is in a connected component with an immunized node. Then, in $G(\s)$ all targeted regions are connected components. Because these regions are targeted, they must have the same size $k$. Also let $\ell$ be the number of target regions. Further, let $Z$ be a component with an immunized node, which has to exist, since profile~$\s$ is non-trivial. Let $u$ be an immunized node in component $Z$ that either is a centroid of $Z$ if an immunized centroid exists in $Z$, or shares an edge with a vulnerable region $M$ that splits $Z$ into connected components of size strictly lower than $|Z|/2$ when removed (which exists by \Cref{lemma:no-centroid-implies-vulnerable}). Next, we deal with the case $k=1$ before continuing with $k > 1$ for the rest of the proof.

    For $k=1$, all targeted regions are isolated nodes. Then, the damage when targeting one of them is $f(1) = f'(0)$ while if there were vulnerable nodes in $Z$, targeting them would be a damage of at least $f'(|Z|-1)$. And as $C_I>1$ and the utility of node~$u$ is positive, it must hold that $|Z|>1$. Hence, there cannot exist a vulnerable node in $Z$ or it would be targeted instead of the isolated ones.
    In this setting, every targeted (and isolated) node would prefer to be part of $Z$, i.e., buying an edge to node~$u$ and immunize. Indeed as $|Z| \ge 2$, some immunized agent in $Z$ bought an edge, and still has positive utility. This implies that $|Z| \ge C_E + C_I$. So the utility after the deviation of an isolated node would be at least $1$, which is more than before. Therefore, profile~$\s$ is not a Nash equilibrium in this case.

    For the rest of the proof, we assume that all of the targeted components have $k > 1$ nodes. 
    \begin{figure}[h]
    \centering
    \includegraphics[width=\linewidth]{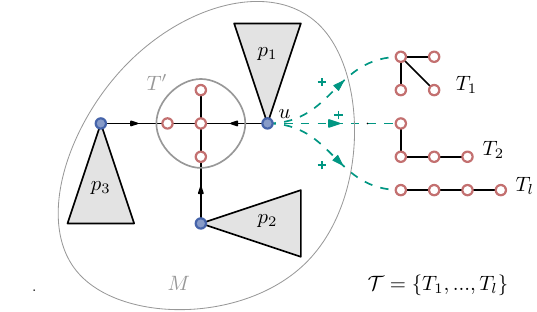}
    \caption{Illustration of the proof of \Cref{lemma:targeted-with-immunized}. An example of the case $k>1$, here with $k=4$ and $v=u$.}
    \label{fig:targeted_immunized}
    \end{figure}
    We will show that for every immunized node~$v \in Z$, there is a vulnerable region $T' \subsetneq Z$ such that $CC_v(\s, T') < |Z|-k$.
    Based on that and previous observations, the region $T'$ that is associated with the immunized node $u$ is targeted in $G(\s)$, which is a contradiction for $G(\s)$ having only complete connected components has targeted regions.

    \begin{step}\label{step1:targeted-with-immunized}
        For every vulnerable node $v \in Z$, there is a vulnerable region $T'$ in $Z$ such that $CC_v(T',\s) < |Z|-k$.
    \end{step}
    \begin{claimproof}[\Cref{step1:targeted-with-immunized}]
        Consider the deviation of agent~$v$ buying an edge to any node in every targeted component, as shown in \Cref{fig:targeted_immunized}, and let $\s'$ denote the obtained strategy profile.

        We will show that the claim of Step 1 will hold for some component~$T' \in \T(\s')$ for which we will argue that it lies in component~$Z$. 
        
        To get an upper bound on $\E_{\T(\s')}[CC_u(\s')]$, we use that strategy profile~$\s$ is a Nash equilibrium and thus
        \begin{equation}\label{eq1:targeted-with-immunized}
            \E_{\T(\s')}[CC_u(\s')] - l C_E \leq \E_{\T(\s)}[CC_u(\s)].
        \end{equation}
        Additionally, consider any leaf $x$ in any targeted component in $T \in \T(\s)$ and its only neighbor~$y$. We claim that agent~$x$ buys the edge~$\{x,y\}$. Otherwise, if agent~$y$ bought the edge $\{x,y\}$, it would be a strict improvement for agent~$y$ to drop the edge since $\left(1-\frac{1}{\ell}\right)k - C_E < k -1$. Because agent~$x$ weakly prefers to keep the edge $\{x,y\}$, it means that $\left(1-\frac{1}{|\T(\s)|}\right) k - C_E  \ge 1$ or equivalently
        \begin{equation}\label{eq2:targeted-with-immunized}
            \ell C_E \leq (\ell-1)k - \ell.
        \end{equation}
        \Cref{eq1:targeted-with-immunized,eq2:targeted-with-immunized} together with $\E_{\T(\s)}[CC_u(\s)] = |Z|$  yield the upper bound of
            $$\E_{\T(\s')}[CC_u(\s')] \leq |Z| + (\ell-1)k - \ell.$$
        This implies that there is a vulnerable region $T'\in \T(\s')$ such that $CC_v(T',\s') < |Z| + (\ell-1)k$, because $\ell$ is strictly positive.

        If $T'$ neither is in $Z$ nor in $T_1,\dots,T_\ell$, then $CC_v(T',\s') = |Z|+\ell k$. If $T'$ is equal to some $T_i$, then $CC_v(T',\s') = |Z|+(\ell-1)k$. Thus, the only remaining possibility is that region~$T'$ is in $Z$. Hence, since agent~$v$ is still connected to the components $T_1,\dots,T_\ell$ in profile~$\s'$, when $T'$ is removed, we have that $$CC_v(T',\s) = CC_v(T',\s') - \ell k < |Z| - k.$$ The second inequality follows from the definition of $T'$. 
    \end{claimproof}

    \begin{step}\label{step2:targeted-with-immunized}
        Some targeted region in $Z$ is targeted in $G(\s)$ by the \sqda{} opponent.
    \end{step}
    \begin{claimproof}[\Cref{step2:targeted-with-immunized}]
       First we want to generalize \Cref{step1:targeted-with-immunized} by finding one region in $Z$ that works for every immunized node. This will either be the region~$T_u'$ associated to node~$u$ by $\Cref{step1:targeted-with-immunized}$ or $M$ if $u$ is not a centroid.
        \begin{enumerate}
            \item If node~$u$ is a centroid, then when removing $T'_u$, the size of the connected component of $u$ is greater or equal than $|Z|/2$ because of the centroid property, and strictly lower than $|Z|-k$ by definition of $T'$. This means that every other remaining component has a size lower or equal to $|Z|/2$ which is strictly lower than $|Z|-k$.
            \item If node~$u$ is not a centroid and the vulnerable region $M$ is different from $T'$, then $CC_u(T'_u,\s) < |Z|-k$, by definition of $T'_u$, but this component includes $M$ which contains a centroid. Hence $|Z|/2 < |Z| - k$. This mean that when removing $M$, every remaining connected component has a size lower than $|Z|/2$ which is lower than $|Z|-k$.
            \item If node~$u$ is not a centroid and the region $M$ is equal to region $T'_u$, then take any immunized node $v$ in a remaining connected component after removing $T'_u$ that shares an edge with a node in $T'_u$. By \Cref{step1:targeted-with-immunized}, there is a vulnerable region $T'_v$ such that $CC_v(T'_v,\s)<|Z|-k$ when removing $T'_v$. If $T'' \neq T'_u$, then $CC_v(T'_v,\s) \ge |Z|/2$ as $v$ is connected to the centroid in $M$. Hence, $|Z|/2 < |Z| -k$ which implies that $CC_v(T'_u,\s) < |Z|-k$.
        \end{enumerate}

        \noindent
        Therefore, there is a region $T'$ (that is either equal to $T'_u $ or to $M$) such that when removed, the size of every remaining connected component is strictly less than $|Z|-k$. Let $p_1,...,p_r$ be these sizes.
        Therefore, $p_i < |Z| - k$ for every $1 \le i \le r$, and $|T'| + \sum_{i=1}^r p_i = |Z|$. This gives us for every $1 \le j \le r$, that 
        \begin{equation}\label{eq1:targeted-with-immunized}
            |T'| + \sum_{i \neq j} p_i > k
        \end{equation}

    \noindent
    Now we are ready to show that $T' \in \T(\s)$, which will be a contradiction. For that, take any $T \in \T(\s)$. We have 
    \begin{align*}
        &U_{f}(T', \s) - U_{f}(T,\s)\\
        =\;&\sum_{i = 1}^r f(p_i) + f(k) - f\left(\sum_{i = 1}^r p_i+ |T'|\right) \\
        =\;& \sum_{i = 1}^r \frac{f(p_i)}{p_i^2}p_i^2 + \frac{f(k)}{k^2}k^2\\
            &\qquad-\; \frac{f(\sum_{i = 1}^r p_i+ |T'|)}{(\sum_{i = 1}^r p_i+ |T'|)^2}\left(\sum_{i = 1}^r p_i+ |T'|\right)^2 \\
        \le\;& \frac{f(\sum_{i = 1}^r p_i+ |T'|)}{(\sum_{i = 1}^r p_i+ |T'|)^2}\left(\sum_{i = 1}^r p^2 + k^2\right)\\
            &\qquad-\; \frac{f(\sum_{i = 1}^r p_i+ |T'|)}{(\sum_{i = 1}^r p_i+ |T'|)^2}\left(\sum_{i = 1}^r p_i+ |T'|\right)^2 \\
    \end{align*}
    This last step used the non-decreasing property of $f/x^2$ and the fact that $|Z| \ge |T_1|$. Now for more clarity define
        $$c\coloneqq \frac{f(\sum_{i = 1}^r p_i+ |T'|)}{(\sum_{i = 1}^r p_i+ |T'|)^2}.$$
    Because $|Z| \ge 2$ (because it contains at least one immunized node and one vulnerable node) and $f$ is strictly convex, we have $f(|Z|) > 0$ and therefore $c > 0$.
    \begin{align*}
        &\frac{1}{c}(U_{f}(T', \s) - U_{f}(T,\s))\\
        \le\;& \sum_{i = 1}^r p^2 + k^2 - \left(\sum_{i = 1}^r p_i+ |T'|\right)^2  \\
        =\;& k^2 - |T'|^2 - \sum_{i=1}^r \Big[p_i \sum_{j\neq i} p_j \Big] - 2|T'| \sum_{i=1}^r p_i \\
        \overset{\text{Eq.} \ref{eq1:targeted-with-immunized}}{<}\;& k^2 - |T'|^2 - \sum_{i=1}^r p_i(k - |T'|)  - 2|T'| \Big[\sum_{i \neq 1} p_i \Big] \\
        \overset{\text{Eq.} \ref{eq1:targeted-with-immunized}}{<}\;& k^2 - |T'|^2 - (k-|T'|) \Big[  \sum_{i \neq j} p_i\Big]    - 2|T'|(k - |T'|) \\
        \overset{\text{Eq.} \ref{eq1:targeted-with-immunized}}{<}\;& k^2 - |T'|^2 - (k-|T'|)^2 - 2|T'|(k - |T'|) \\
        =\;&  k^2 - |T'|^2 - k^2 + 2k|T'| - |T'|^2 - 2|T'|k + 2|T'|^2 \\
        =\;& 0
    \end{align*}
    Therefore it holds that $U_{f}(T', \s) < U_{f}(T,\s)$, which means that this opponent would target $T'$.
    \end{claimproof} 
    This finishes the proof.
\end{proof}

\noindent 
We continue with the setting without targeted cut-vertices and show the existence of at least two vulnerable nodes.

\lemmaTwoVulnerable*

\begin{proof}
    Assume towards a contradiction, that there are no targeted cut-vertices but there is a connected component $K$ with an immunized agent and targeted region that contains exactly one vulnerable agent~$x$. Note, that $x$ must be targeted, since there is a targeted region in $K$.
    
    The goal is to prove that another agent in component $K$ strictly prefers to stop immunizing. To do so, we find another agent in the same situation as agent~$x$ and because agent~$x$ made this choice, this other agent will also prefer it. 

    First, notice that agent~$x$ can only have one incident edge, i.e., it is a leaf. Indeed, assuming it has at least two incident edges, because node~$x$ is not a cut-vertex, another path of immunized nodes link the two endpoints of these edges. Thus, one of them is redundant against any well-behaved opponent.
    Therefore, there is exactly one immunized node~$u$ that is linked to node~$x$. 
    
    Since $x$ is a leaf and since $C_E>1$, the edge $\{x,u\}$ must be bought by agent~$x$. Since this edge must be profitable for agent~$x$, it holds that $|K| \ge C_E + 1 > 2$ so there is at least one other immunized agent in $K$, other than agent $u$. Observe that the immunized component of node~$u$ forms a tree~$T_u$, since non-tree edges can be dropped without changing the connectivity. 
    
    Now, root $T_u$ at node~$u$ and consider a leaf $v$. Since node~$x$ has only one incident edge to node $u$, it cannot share an edge with $v$. Also, note that $v\neq x$, since node $v$ is immunized. Now, note that removing $x$ or $v$ lowers $U_f$ by the same amount: $f'(|K|-1)$. Hence selling immunity makes agent~$v$ targeted, without removing any targeted region.

    Now because agent~$x$ chooses not to immunize in Nash equilibrium~$\s$, we have $\frac{1}{|\T|}|K| \le C_I$. This means that 
        \begin{equation}\frac{1}{|\T(\s)| + 1}\E_{\T(\s)}[CC_v(\s)] \le \frac{1}{\T(\s) + 1} |K| < C_I. \label{eq1:lemma_two_vulnerable}\end{equation}
    However, consider the deviation when agent~$v$ does not immunize and let $\s'$ be the corresponding strategy profile. Then using conditional expectations:
    \begin{align*}
        \E_{\T(\s) \cup \{v\}}[CC_v(\s')] =\;& \frac{|\T(\s)|}{|\T(\s)| +1} \E_{\T(\s)}[CC_v(\s')]\\
        =\;& \frac{|\T(\s)|}{|\T(\s)| +1} \E_{\T(\s)}[CC_v(\s)].
    \end{align*}
    This last equality comes from the fact that node~$v$ is not in any region of $\T(\s)$ and is not linked to any vulnerable node. So being vulnerable does not change the size of any connected component when any region in $\T(\s)$ is attacked. Now because of the inequalities we established:
    \begin{align*}
        &\frac{|\T(\s)|}{|\T(\s)| +1} \E_{\T(\s)}[CC_v(\s)] \\
        =\;& \E_{\T(\s)}[CC_v(\s)] - \frac{1}{\T(\s) + 1}\E_{\T(\s)}[CC_v(\s)] \\
        >\;& \E_{\T(\s)}[CC_v(\s)] - C_I,
    \end{align*}
    where the inequality follows from Inequality~(\ref{eq1:lemma_two_vulnerable}). This allows us to conclude that 
    \begin{align*}
        u_v(\s') =\;& \E_{\T(\s) \cup \{v\}}[CC_v(\s')] - C_E\\
        >\;& \E_{\T(\s)}[CC_v(\s)] - C_E - C_I\\
        =\;& u_v(\s).
    \end{align*}
    Therefore strategy profile~$\s$ is not a Nash equilibrium, which gives us a contradiction.
\end{proof}

\noindent 
Finally, we can show that most one connected component with a targeted regions exists.

\lemmaSoloImmunizedTargeted*
\begin{proof}
    Assume there are two components $K_1$ and $K_2$ with immunized nodes $u_1$ and $u_2$, respectively, and both have targeted regions. These targeted regions are singletons, by \Cref{lemma:targetsingletons}, and are not cut-vertices by assumption. Therefore, when agent~$u_1$ buys an edge to node~$u_2$, its new utility is at least $|K_1| + |K_2|-1 -C_E - c$ where $c$ denotes the costs of the existing strategy of agent~$u_1$. However, before the deviation this was strictly smaller than $|K_1| - c$ as there are targeted regions in $K_1$. Because $|K_2| > C_E+1$, by \Cref{lemma:minimal_size}, this deviation is strictly improving.
\end{proof}

Now, we derive the important statement, that without targeted cut-vertices no isolated agents exist. This is helpful, since such agents contribute to low social welfare. 

\lemmaNoIsolated*
\begin{proof}
    Let $\s$ be a non-trivial Nash equilibrium. By \Cref{lemma:targeted-with-immunized}, there is a component $K$ with at an immunized agent and, if there are vulnerable nodes in $G(\s)$, a targeted region. By \Cref{lemma:targetsingletons}, targeted regions must be singletons and by assumption there are no targeted cut-vertices. Hence, every vulnerable agent in $K$ is targeted and not a cut-vertex.
    
    Now, towards a contradiction, assume there exists an isolated agent~$x$. If agent~$x$ was already immunized, then buying an edge to any immunized node in $K$ would be a strict utility increase, as there are at least $C_E+1$ nodes in $K$ and no vulnerable cut-vertex exists in $K$. Therefore, agent~$x$ is vulnerable, which means that there exists a targeted region, and thus targeted non-cut-vertices in component~$K$.
    
    Now there will be several cases depending on the structure of component~$K$, see \Cref{fig:no-isolated} for an illustration. 
    \begin{enumerate}
        \item Some agent in $K$ bought two edges.
        \item There is only one immunized agent in $K$.
        \item An immunized agent bought an edge that disconnects $K$ when removed and no agent bought more than one edge.
        \item Every edge bought by an immunized agent in $K$ is part of a cycle, no agent bought more than one edge, and some immunized agent bought an edge.
    \end{enumerate}
    \begin{figure}[h]
        \centering
        \includegraphics[width=\linewidth]{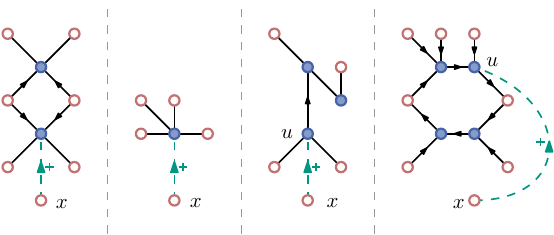}
        \caption{The four cases of the proof. The dashed edge is the proposed deviation.}
        \label{fig:no-isolated}
    \end{figure}

    We still have to prove that those four cases are exhaustive. Indeed, if we are not in the cases $1$ or $2$, then because there are at least $|K| -1$ edges, one of them must be bought by an immunized agent. If one of such edges disconnect $K$ when removed, then we are in Case $3$, and else we are in Case $4$. 
    
    In all of those cases, we will find that buying an edge to an immunized node of $K$ will be a strictly improving deviation for agent $x$. Without buying immunity in cases $1,2$ and $4$ and sometimes buying immunity in Case $3$.

    \begin{case}\label{step1:no-isolated}
        We show that if some agent in $K$ bought two edges, then agent~$x$ strictly prefers to buy an edge to any immunized node of component~$K$.
    \end{case}
    \begin{claimproof}[\Cref{step1:no-isolated}]
        Assume that some agent $v$ bought two edges. If agent~$v$ is immunized then because its utility is positive, we have $$|K| \ge 2 C_E + C_I > 2 C_E + 1.$$ If agent~$v$ is vulnerable, consider the deviation of selling every bought edge. If agent~$v$ becomes isolated by this, then its utility becomes $1$, and else the probability of being destroyed does not increase and was at most $\frac12$, by \Cref{lemma:two-vulnerable}. Hence, the new utility is at least $\frac{1}{2}\cdot 2 \ge 1$. Therefore the Nash property gives us $|K| - 2 C_E \ge 1$. Thus, independent of agent~$v$'s immunization status, we get that $|K| \ge 2 C_E +1$.
        
        Additionally, \Cref{lemma:two-vulnerable} ensures that component~$K$ contains at least two vulnerable agents. Now, with these results, after the deviation of buying an edge to an immunized node of $K$, it holds that the probability of node $x$ being destroyed is at most $\frac{1}{3}$. Therefore, the new utility of agent~$x$ is at least $$\frac{2}{3}(2 C_E+1) \ge  \frac{1}{3}C_E + \frac{2}{3} > 1.$$ This deviation is again a strict increase of utility for $x$, which contradict the fact that $\s$ was a Nash equilibrium.
    \end{claimproof}

    \begin{case}\label{step2:no-isolated}
        We show that if $K$ contains exactly one immunized agent, then agent~$x$ strictly prefers to buy an edge to this node.
    \end{case}
    \begin{claimproof}[\Cref{step2:no-isolated}]
        If there is only one immunized agent in $K$, then buying an edge to it yields an utility of at least
            $$\frac{|K|-1}{|K|} |K| - C_E = |K|-1-C_E$$
        for agent~$x$. Note, that component~$K$ must be a star with the immunized agent as center.
        
        Now, if there were no other targeted regions, then the fact that the vulnerable nodes in $K$ preferred to keep their edges to the star center, means that $$\frac{|K|-2}{|K|-1}(|K|-1)-C_E \ge 1.$$ This implies that $|K|-1 - C_E > 1$ and therefore that this deviation is strictly increasing agent~$x$'s utility.
        
        However, there might be other targeted regions~$T_1,...,T_r$. By \Cref{lemma:solo-immunized-targeted}, we know that they are vulnerable components. 
        In that case, necessarily $$|K| > |T_1| > C_E +1.$$ Indeed, the damage when targeting a node in $K$ is $f'(|K|-1)$ while it is $f(|T_1|)$ when targeting $T_1$. Now, because $$f(|T_1|) = \sum_{i=0}^{|T_1| - 1} f'(i),$$ and, because $|T_1| \ge C_E+1 > 2$, the equality $f'(|K|-1) = f(|T_1|)$ gives us $$f'(|K|-1) \ge f'(|T_1|-1) + f'(1) > f'(|T_1| - 1).$$ We conclude by convexity of $f$.
        
        Because $|K|$ and $|T_1|$ are two integers, it gives us $|K| > C_E + 2$. In that case agent~$x$'s utilitiy after buying the edge  is at least $|K|-1 - C_E > 1$, hence the deviation is always strictly improving.
    \end{claimproof}

    Now for the two last cases, we will first prove general results that will apply to both. As, by \Cref{lemma:two-vulnerable}, there are at least two vulnerable nodes in $K$ and at least $|K|-1$ edges, when no agent bought more than one edge then a vulnerable agent $v$ bought an edge.

    We assume towards a contradiction,that there is at most one targeted region outside of $K$. In this case, agent~$v$'s utility is 
        $$ u_v(\s) < \frac{|\T(\s)|-1}{|\T(\s)|}|K| - C_E,$$
    since if agent~$v$ is not destroyed, then it is connected to at most all the nodes of $K$ and the inequality is strict since there is another targeted node in $K$.  
    
    Now we consider that agent~$v$ drops the bought edge and let $\s'$ denote the corresponding strategy profile. We now show that $u_v(\s') \geq 1$. This holds, since the probability of being destroyed does not increase when dropping edges, and this probability for agent~$v$ is at most $\tfrac12$, because there is at least one other targeted node. This means that if an immunized agent bought an edge to node~$v$, then $u_v(\s') \geq \tfrac{1}{2} \cdot 2$. If, on the other hand, no agent bought an edge to node~$v$, then node~$v$ is isolated in profile~$\s'$ and thus no longer targeted because there is a non-isolated vulnerable nodes in $K \setminus \{v\}$. Hence, in this case we have $u_v(\s') = 1$. 
    
    Hence, since profile $\s$ is a Nash equilibrium, we have $$1 \leq u_v(\s') \leq u_v(\s) < \frac{|\T(\s)|-1}{|\T(\s)|}|K| - C_E.$$ Therefore, if agent~$x$ buys an edge to an immunized node of $K$, its new utility is at least $\frac{|\T(\s)|-1}{|\T(\s)|}|K| - C_E > 1$.

    Now, consider that in profile~$\s$ agent~$x$ buys an edge to some immunized node in $K$ and let $\s''$ be the corresponding strategy profile. In this case, we have $|\T(\s)| \leq |\T(\s'')|$, since if there is no targeted region outside of $K$, then node $x$ becomes a new targeted region and all other targeted regions stay targeted. If there is a targeted region $T$ outside of $K$, then in $\s''$ region $T$ is no longer targeted, since the size of component $K$ increased by one node and thus, the \sqda{} would only attack vulnerable nodes in $K$. Therefore, in profile $\s''$ agent~$x$'s utility is at least $$\frac{|\T(\s'')|-1}{|\T(\s'')|}|K| - C_E \geq \frac{|\T(\s)|-1}{|\T(\s)|}|K| - C_E > 1.$$ Thus, since in profile~$\s$ agent~$x$ has utility $1$, this deviation strictly improves agent~$x$'s utility and hence contradicts that profile~$\s$ is a Nash equilibrium.     

    Thus, there are at least two targeted regions outside of $K$. Notice that they cannot be in components with immunized nodes, by \Cref{lemma:solo-immunized-targeted}. Therefore, there must exist whole vulnerable components $T_1,...,T_r$, which have, by \Cref{lemma:minimal_size}, size $|T_1| \ge C_E +1$. 
    Now we will look at whether $K$ contains an edge bought by an immunized agent that is not part of a cycle or if all such edges are.

    \begin{case}\label{step3:no-isolated}
        We show that if an immunized agent $u$ bought an edge $e$ that disconnects $K$ when removed and no agent in $K$ bought more than one edge, then buying an edge to node~$u$ and maybe immunizing is strictly improving for agent~$x$.
    \end{case}
    \begin{claimproof}[\Cref{step3:no-isolated}]
        Look at the deviation of agent~$u$ selling immunity and selling edge~$e$. By assumption, this splits $K$ in two components $K_u$ and $K_v$. If node~$u$ is not targeted after the deviation, then its new utility is at least $1$ which means that $|K| > C_E + C_I + 1$ and agent~$x$ strictly prefers to immunize and buy an edge to node~$u$.

        If node~$u$ becomes targeted after this deviation, it means that $|K_u| \ge |T_1| \ge C_E +1$. However, as agent~$u$ bought every edge to $K_v$ and there are targeted regions outside of $K$, we can apply \Cref{cor:size-after-edge} to conclude that $|K_v| \ge C_E$. Hence, we get $|K| \ge 2 C_E +1$. Now if agent~$x$ buys an edge to node~$u$, its new utility is at least $$\frac{2}{3} (2C_E +1) - C_E = \frac{1}{3}C_E + \frac{2}{3} > 1,$$ which is strictly improving.
    \end{claimproof}

    \begin{case}\label{step4:no-isolated}
        We show that if all edges that are bought by an immunized agent of $K$ are part of a cycle, no vertices bought more than one edge, and some immunized agent bought an edge, then buying an edge to an immunized node of $K$ is strictly improving for agent~$x$.
    \end{case}
    \begin{claimproof}[\Cref{step4:no-isolated}]
        In this case, we assume that component $K$ contains a cycle. First of all, we prove that this cycle is unique. For this, construct a directed graph $\vec{G}(\s) = (K, \vec{E}(\s))$ based on the strategy profile $\s$ by taking the nodes of $K$, and such that every edge $(u,v) \in \vec{E}(\s)$ indicates that the node $u\in K$ buys the edge $\{u,v\} \in E(\s)$. The assumption that no agent buys more than one edge gives us that every node in $\vec{G}$ has out-degree one. Since $K$ is connected, graph~$\vec{G}$ is weakly connected.  For such a directed and weakly connected graph it is known that it has at most one cycle and that this cycle is directed. 

        Therefore, we know that component $K$ consists of a directed cycle $C$ containing all immunized nodes (by assumption) as well as trees attached to the nodes of cycle $C$. Further, these trees are only attached to immunized nodes and only consist of one node (they are a dangling leaf). This is because, otherwise, there was a vulnerable region of size at least $2$ which would then be targeted and would thus contradict \Cref{lemma:targetsingletons}. For the same reason, we know that no two adjacent nodes on the cycle $C$ are vulnerable. Additionally, we claim that the cycle~$C$ contains at least two vulnerable nodes. This holds, since if the cycle~$C$ only contains immunized nodes, then one of them would prefer to drop its cycle edge to safe the edge cost without losing connectivity. Also, if there was only one vulnerable node in cycle~$C$, then the immunized agent that bought an edge to this node would drop its edge for the same reason. An example of such a network is shown in \Cref{fig:no-isolated}. Note that all leaves in component~$K$ bought the edge to their neighbor in cycle~$C$, as otherwise their neighbor would prefer to drop the edge.

        On cycle~$C$, consider the immunized regions $I_1,\dots,I_k$. From previous considerations, we know that $k \geq 2$ and that all these immunized regions are directed paths. In the following let $s_i$ for $i \in [k]$ be the size of the immunized region $I_i$ plus the number of dangling leaves that bought an edge to $I_i$. With this, we have that
\begin{equation}\label{eq0:step4:no-isolated}
            |K| = k + \sum_{i=1}^k s_i
        \end{equation}
        since between two immunized regions on the cycle~$C$ there is exactly one vulnerable node.
        Without loss of generality, we will assume in the following that the indices of the immunized regions $I_1,\dots,I_k$ are arranged such that $s_1 \geq \dots \geq s_k$. Having this in mind, let $u$ be the node in the largest immunized region $I_1$ (according to the $s_i$) that bought an edge to a vulnerable node in the cycle~$C$.
        In the remainder of this proof, we will show that component $K$ is large enough such that the isolated agent~$x$ would prefer to buy an edge to node~$u$, which would be contradiction to strategy profile $\s$ being a Nash equilibrium.

        To get this lower bound on the size of $K$, consider the deviation of agent~$u$ dropping its edge~$e$, and buying edges to any node of the targeted components $T_1$ and $T_2$, that are outside of $K$, instead. Let $\s'$ denote the obtained strategy profile. Since strategy profile $\s$ is a Nash equilibrium, we know that
        \begin{equation}\label{eq1:step4:no-isolated}
            \E_{\T(\s')}[CC_u(\s')] \le \E_{\T(\s)}[CC_u(\s)] + C_E.
        \end{equation}
        Since there are targeted regions inside $K$ for profile~$\s$, we also know that $\E_{\T(\s)}[CC_u(\s)] < |K|$. Together with \Cref{eq1:step4:no-isolated}, this implies that there is a vulnerable region $T' \in \T(\s')$ such that 
        \begin{equation}\label{eq2:step4:no-isolated}
            CC_u(T',\s') < |K| + C_E.
        \end{equation}
        We argue, that this targeted region $T'$ has to be in cycle~$C$, since $CC_u(T',\s') \geq |K| + C_E$ holds for all other possibilities, which we will now show.
        
        If $T'$ neither is in $K$ nor in the targeted components $T_1$ or $T_2$, then 
            $$CC_u(T',\s') = |K| + 2|T_1| \ge |K| + 2C_E + 2,$$
        where $|T_1| \geq C_E+1$ holds by \Cref{lemma:minimal_size}.
            
        If $T'$ is one of the targeted components $T_1$ or $T_2$, then 
            $$CC_u(T',\s') = |K| + |T_1| \ge |K| + C_E + 1$$
        as agent~$u$ is still connected to $K$ and at least one of $T_1$ and $T_2$ (that have the same size since they are both targeted).
        Thus, the only remaining possibility is that the vulnerable region $T'$ is in component~$K$. If $T'$ was one of the dangling leaves, then 
            $$CC_u(T',\s') = |K|-1 + 2|T_1| \ge |K| + C_E + 1,$$
        as only one node is destroyed that is not a cut-vertex. Therefore, region~$T'$ has to be a node in the cycle~$C$.
        Hence, since agent~$u$ is still connected to the components $T_1$, $T_2$ and the $s_1$ nodes from its immunized component (with the dangling leaves attached to it) in $G(\s')$, when $T'$ is removed, we have that
        \begin{align}
            s_1 \le\;& CC_u(T',\s') - 2|T_1|\nonumber\\
            <\;& |K| + C_E - 2(C_E +1).\nonumber\\
            \leq\;& |K| - (C_E + 2),\label{eq3:step4:no-isolated}
        \end{align}
        Where we used \Cref{eq2:step4:no-isolated} and $|T_1| \geq C_E + 1$ for the second inequality.
        By rearranging Inequality~(\ref{eq3:step4:no-isolated}), we get
        
        \begin{equation}\label{eq4:step4:no-isolated}
            C_E + 2 < |K| - s_1 = \sum_{i=2}^ks_i + k,
        \end{equation}
        where the last equality is due to \Cref{eq0:step4:no-isolated}.
        To get rid of the sum, note that $s_1$ (as the largest $s_i$) is larger or equal than the average value
            $$\overline{s} \coloneqq \frac{1}{k}\sum_{i=1}^ks_i = \frac{1}{k}(|K|-k)$$
        of all $s_i$ values, and thus
        \begin{align*}
            \sum_{i=2}^ks_i =\;& |K|- k - s_1 \\
            \leq\;& |K| - k - \overline{s}\\
            =\;& |K| - k - \frac{1}{k}(|K|-k)\\
            =\;& \frac{k-1}{k}(|K|-k).
        \end{align*}
        Plugging this into \Cref{eq4:step4:no-isolated} yields:
        \begin{align*}
            C_E + 2 < \frac{k-1}{k}(|K|-k) + k = \frac{k-1}{k}|K| + 1.
        \end{align*}
        Which gives the following lower bound to the size of component $K$: 
        \begin{equation}\label{eq5:step4:no-isolated}
            |K| > \frac{k}{k-1}(C_E + 1).
        \end{equation}

        Now, we show that this lower bound to the size of component $K$ is large enough such that the isolated agent~$x$ strictly prefers to buy an edge to node~$u$. Let $\s''$ be the strategy profile obtained by agent~$x$ buying the edge $\{x,u\}$. We show that the utility of agent~$x$ is strictly greater than $1$ in profile~$\s''$ and thus strictly greater as in the given Nash equilibrium~$\s$. 
        With node~$x$ being part of component~$K$ in $G(\s'')$, all targeted regions in $K$ (w.r.t. profile $\s$) now cause even more damage because of the additional node~$x$. Thus, all vulnerable regions outside of $K$ that were targeted in profile~$\s$ are no longer targeted in profile~$\s''$. Since for all but one  of them (namely the targeted region consisting of agent~$x$), the node~$x$ remains in a component of at least $|K|$ nodes, we get that the utility of agent~$x$ in $\s''$ is
            $$u_x(\s'') = \frac{|\T(\s'')|-1}{|\T(\s'')|}|K| - C_E.$$
        Since $\frac{x-1}{x}$ is strictly increasing for $x > 0$, we can use a lower bound on the number of targeted regions in $G(\s'')$ to bound the utility $u_x(\s'')$ from below. For this, note that there are at least $k+1$ vulnerable regions in $K \cup \{x\}$ in $G(\s'')$, namely node~$x$ and the nodes separating the immunized regions $I_1,\dots,I_k$ on the cycle~$C$. We know that all of them are targeted, since they cause the same damage. Therefore, we get that $k+1 \leq |\T(\s'')|$ and thus
        \begin{align*}
            u_x(\s'') \geq \frac{k}{k+1}|K| - C_E > \frac{k-1}{k}|K| - C_E,
        \end{align*}
        where we used that $\frac{x-1}{x}$ is strictly increasing for $x > 0$ for the second inequality. Now we can use the obtained lower on the size of component~$K$ from \Cref{eq5:step4:no-isolated} to get that
        \begin{align*}
            u_x(\s'') >\;& \frac{k-1}{k}|K| - C_E\\
            >\;& \frac{k-1}{k}\left(\frac{k}{k-1}(C_E+1)\right) - C_E\\
            =\;& 1 = u_x(\s),
        \end{align*}
        which finally shows that profile~$\s$ is not a Nash equilibrium and thus concludes the proof by contradiction.
    \end{claimproof}
    This concludes the proof.
\end{proof}

\section{Omitted Details from the Analysis for the Tailored Opponent}
As contrast for our positive results on the maximum carnage, random attack, and the \sqda{} opponents, we show that there exist opponents, such that Nash equilibria can have social welfare of $\Theta(n)$, i.e., the lowest possible welfare. This shows the counter-intuitive result, that opponents exist, that achieve a lower social welfare than the opponent that actually aims for minimizing the social welfare.

\thmBadOpponent*
\begin{proof}
We consider a family of instances with $n$ agents and $C_E= C_I = 6$ and a tailored $f$-opponent $\mathcal{A}$, where $f$ is defined as follows: 
$$f(i) = \begin{cases}
           2, & \text{if $i=1$},\\
           3, & \text{if $i=7$},\\
           5, & \text{if $i=8$},\\
           4, & \text{if $i=9$},\\
           7, & \text{if $i=10$},\\
           0, &\text{otherwise}.
\end{cases}
$$

For every $n \geq 10$, we consider the strategy profile~$\s_{\text{bad}}$ illustrated in \Cref{fig:bad-Nash_appendix}. There, all but nine agents are isolated.
\begin{figure}[h]
    \centering
    \includegraphics[width=0.8\linewidth]{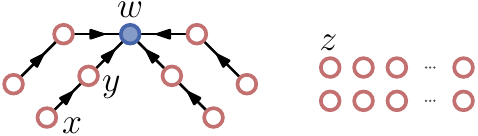}
        \captionof{figure}{The Nash equilibrium~$\s_{\text{bad}}$ with social welfare in $\Theta(n)$, for our tailored opponent, with $C_E = C_I = 6$.}
        \label{fig:bad-Nash_appendix}
\end{figure}

 Strategy profile~$\s_{\text{bad}}$ defines a network $G(\s_{\text{bad}})$ that is a star with four rays, where every ray consists of two vulnerable nodes. Besides this, there are $n-9$ isolated nodes. All nodes except the star center are vulnerable nodes. All edges are bought be the respective agent that is further away from the star center. There are therefore $n-5$ vulnerable regions, and $n-9$ targeted nodes. 
 
 Note that, one of the isolated nodes will be targeted. Thus, the social welfare of profile~$\sbad$ is $$(n-9)\cdot \frac{n-10}{n-9}\cdot 1 + 9\cdot 9 - 8 C_E - C_I \in \Theta(n),$$ since each of the $n-9$ isolated nodes survives with probability $\frac{n-10}{n-9}$ and since the nine nodes in the big connected component are not targeted, each of them has connectivity~$9$.   

Now we prove that profile $\sbad$ indeed is a Nash equilibrium. There are only four different classes of equivalent agents. But first, notice that if any agent buys $k$ edges, the size of its connected component will be lower than $9 + k$, which is strictly lower than $k \cdot C_E$ for $k \ge 2$. Therefore any best response of any agent can buy at most one edge.

\textbf{Node $w$:} Can stop immunizing, stop immunizing and buy an edge, or only buy an edge. In the first two cases, agent~$w$ becomes always destroyed, which is not improving as its current utility is $9 - C_I  = 3 > 0$. If agent~$w$ buys an edge to its own component or to an isolated node, this lowers its utility by $C_E = 6$.

\textbf{Node $y$:} currently has a utility of $9 - C_E = 3$. Agent~$y$ can immunize, immunize and sell its edge, immunize and swap its edge, only sell its edge, or only swap its edge.
        \begin{itemize}
            \item Immunizing without selling the edge yields additional costs of $C_I = 6$, so that agent~$y$'s utility is negative.
            \item Immunizing and selling the edge gives a utility of $-4$.
            \item Only selling the edge leads to a utility of $2$ which is less than the current one.
            \item Swapping the edge to another node in the same component does not change the utility, swapping it for an isolated node leads to an utility of $-3$. 
        \end{itemize}

    \textbf{Node $x$:} Has the same utility as node $y$ and can do the same strategy changes as node $y$ with the same utility changes. Additionally node $x$ could swap its edge to node $w$, without a change in utility.

    \textbf{Node $z$:} The current utility of agent~$z$ is $\frac{n-10}{n-9}>0$. It can immunize, or immunize and buy an edge, or only buy an edge.
        \begin{itemize}
            \item If agent~$z$ immunizes, the new utility will be $-5$ which is less than the current one.
            \item If agent~$z$ immunizes and buys an edge, the cost of $12$ for the edge and immunization will be higher than the size of its connected component.
            \item If agent~$z$ buys an edge to another isolated node, the new utility is $-4$. If it buys the edge to node~$w$, node $z$ gets destroyed and thus has a utility of $-6$. The same happens for buying an edge to node $x$ or $y$.
        \end{itemize}
Thus, in profile $\sbad$ all agents play best response.
\end{proof}

\end{document}